\documentclass[letterpaper]{article}
\usepackage{aaai20}
\usepackage{times}
  \usepackage{helvet}
  \usepackage{courier}
  \usepackage[hyphens]{url}
  \usepackage{graphicx}
  \usepackage{graphicx}
\usepackage[utf8]{inputenc}
\makeatletter%
\@ifclassloaded{acmart}{
  \makeatletter
  \def\mdseries@tt{m}
  \makeatother
}{}
\makeatother

\usepackage{amsthm}
\usepackage{thmtools}
\declaretheorem{theorem}
\declaretheorem{corollary}

\usepackage[table,dvipsnames]{xcolor}
\usepackage{epsfig}
\usepackage{wrapfig}
\usepackage{pgfplotstable}
\usepackage{pgfplots}
\usepgfplotslibrary{groupplots}
\pgfplotsset{compat=newest} %
\usepackage{amsmath,mathtools}
\usepackage{amssymb}
\usepackage{relsize}
\usepackage{multicol}
\usepackage{bbm}
\usepackage{float}
\usepackage{hyperref}
\usepackage{microtype}
\usepackage{tikz}
\usetikzlibrary{decorations.text,calc,shapes,arrows,arrows.meta, positioning,shapes.misc,decorations.markings,decorations.markings,decorations.pathreplacing,matrix,spy}
\usepackage{rotating}

\usepackage{booktabs}
\usepackage{multirow}
\usepackage{adjustbox}
\usepackage{verbatim}
\usepackage[T1]{fontenc}
\usepackage[title]{appendix}
\usepackage{adjustbox} %
\usepackage{graphicx}
\usepackage{caption}
\usepackage{subcaption}
\usepackage[noend]{algpseudocode}
\usepackage{algorithm}
\usepackage{siunitx}
\usepackage{nicefrac}
\usepackage[colorinlistoftodos]{todonotes}
\usepackage{spverbatim}
\usepackage{seqsplit}

\usepackage{enumitem}
\setlist{nosep} %

\usepackage{array}
\newcolumntype{R}[2]{%
    >{\adjustbox{angle=#1,lap=\width-(#2)}\bgroup}%
    l%
    <{\egroup}%
}

\usepackage{natbib}
\setcitestyle{authoryear,open={(},close={)}}
\renewcommand*{\cite}{\citep}
\newcommand{\textcite}[1]{\citet{#1}}
\usepackage[autostyle, english=american]{csquotes}
\MakeOuterQuote{"}

\makeatletter
\pgfplotsset{
    every axis x label/.append style={
        alias=current axis xlabel
    },
    legend pos/outer south/.style={
        /pgfplots/legend style={
            at={%
                (%
                \@ifundefined{pgf@sh@ns@current axis xlabel}%
                {xticklabel cs:0.5}%
                {current axis xlabel.south}%
                )%
            },
            anchor=north
        }
    }
}
\makeatother

\newcolumntype{t}{>{\ttfamily}l}
\newcolumntype{T}{>{\ttfamily}c}

\newcolumntype{$}{>{\global\let\currentrowstyle\relax}}
\newcolumntype{^}{>{\currentrowstyle}}

\everypar{\looseness=-1} %
\newcommand{\bigO}{\ensuremath\mathcal{O}}
\hypersetup{draft}

\begin{document}

\title{A New Burrows Wheeler Transform Markov Distance}

\author{Edward Raff\textsuperscript{1,2,3} \and Charles Nicholas\textsuperscript{3} \and Mark McLean\textsuperscript{1} \\
\textsuperscript{1}{Laboratory for Physical Sciences}, \textsuperscript{2}{Booz Allen Hamilton}, \textsuperscript{3}{University of Maryland, Baltimore County}\\
raff\_edward@bah.com, nicholas@umbc.edu, mclean@lps.umd.edu}

\maketitle

\begin{abstract}
Prior work inspired by compression algorithms has described how the Burrows Wheeler Transform can be used to create a distance measure for bioinformatics problems. We describe issues with this approach that were not widely known, and introduce our new Burrows Wheeler Markov Distance (BWMD) as an alternative. The BWMD avoids the shortcomings of earlier efforts, and allows us to tackle problems in variable length DNA sequence clustering.%
BWMD is also more adaptable to other domains, which we demonstrate on malware classification tasks. 
Unlike other compression-based distance metrics known to us, BWMD works by embedding sequences into a fixed-length feature vector. This allows us to provide significantly improved clustering performance on larger malware corpora, a weakness of prior methods. 
\end{abstract}

\section{Introduction}

Compression algorithms can be used to measure the similarity between arbitrary sequences with little required domain knowledge or expertise. They have been used in bioinformatics\cite{Mantaci2008}, time series classification and clustering\cite{Keogh:2004:TPD:1014052.1014077}, and malware analysis \cite{Borbely2015}.  The bioinformatics and malware analysis domains can be particularly attractive for compression-based similarity measures. Both of these domains involve "short" sequences of tens of thousands of steps, and can often reach $10^8$ steps in length.  Other machine learning techniques often fail to work when dealing with sequences of such variety and length. 

In this work, we note that the Extended Burrows Wheeler Transform (EBWT) \cite{Mantaci:2005:EBW:2134629.2134645} is a compression-based distance metric designed explicitly around the Burrows Wheeler Transform (BWT) \cite{Burrows1994} algorithm for use in bioinformatics. While EBWT has been useful in that domain, we have discovered a number of weaknesses in this method that reduce its effectiveness and prevent it from being useful in other domains, such as malware detection. 

To remedy these issues, we develop a new BWT-inspired distance measure that we refer to as the Burrows Wheeler Markov Distance (BWMD). Unlike EBWT, BWMD is a valid\footnote{A distance metric is considered true, or valid, if it adheres to the properties of reflexivity, symmetry, and triangularity.} distance metric, and can scale to far larger problems that EBWT cannot tackle due to computational limits.
Compared to other compression-based distances, our BWMD is the first to work by embedding a sequence into a Euclidean vector space. This gives a significant advantage to our approach in terms of clustering and query speed. This advantage is achieved by using algorithms that are designed around Euclidean distance, like k-means, that other methods cannot 
leverage.

We will begin by reviewing related work in the compression distance space, and the needed details of the BWT, the prior method EBWT, and a related method known as LZJD, in \autoref{sec:compression_dists}. 
Next we will begin with a description of the new BWMD in \autoref{sec:bwmd}. In \autoref{sec:theoretical_results} we will develop a number of new theoretical insights, proving 1) how EBWT has dramatic failure cases that violate our intuition of how a distance measure should work, 2) that BWMD does not have these failure cases, and 3) comparing how EBWT, BWMD, and LZJD handle randomness, and 4) that BWMD has unique properties in this regard. We will then move into empirical results in \autoref{sec:dna} by comparing BWMD with EBWT on DNA sequence clustering, where we show that BWMD is able to cluster DNA sequences of varying lengths that EBWT fails to cluster in a meaningful way. In \autoref{sec:malware} we will show how BWMD is able to scale to malware classification and clustering tasks that are beyond EBWT's computational ability. Though LZJD provides better classification accuracy at this task, BWMD provides superior clustering results. Finally we will conclude in \autoref{sec:conclusion}

\section{Compression Distances} \label{sec:compression_dists}

Compression in general can be seen as one way of performing many machine learning tasks, and has deep connections to statistical methods. Following this intuition, \textcite{Li2004} introduced the Normalized Information Distance (NID) as a method of measuring similarity using compression. Given a function $K(x)$ that computes the Kolmogorov complexity (i.e., return the length of the shortest computer program that produces $x$ as output), and the associated conditional Kolmogorov complexity $K(x | y)$ (i.e., the length of the shortest computer program that produces $x$ as output given $y$ as input), the NID is a metric as defined in \eqref{eq:nid}.
The Kolmogorov complexities are uncomputable functions, making NID of no practical use. 

\begin{equation} \label{eq:nid}
\text{NID}(x,y) = \frac{\max\left(K(x | y), K(y | x)\right)}{\max\left(K(x), K(y)\right)}
\end{equation}

To remedy this situation,  \textcite{Li2004} went on to present the Normalized Compression Distance (NCD), which replaces the uncomputable $K(x)$ with $C(x)$, which returns the length of $x$ in bytes after running a compression algorithm. To approximate $K(x | y)$, the concatenation of $x$ and $y$ (denoted by $x\Vert y$) is used, giving the final form of NCD in \eqref{eq:ncd}.  The compression algorithm chosen impacts the quality of the results. The bzip and LZMA algorithms have been popular due to a combination of reasonable run time performance and generally satisfactory compression ratios. 

\begin{equation}\label{eq:ncd}
    \text{NCD}(x,y) = \frac { C ( x \Vert  y ) - \min \left( C ( x ) , C ( y ) \right) } { \max \left( C ( x ) , C ( y ) \right) }
\end{equation}

Where $C(x)$ is the (integer) length of object $x$ when compressed.
NCD, and compression-based distances in general, do not require significant feature engineering to be applied in practice. This has made them popular for genomic phylogeny \cite{Li2004,Cilibrasi:2005:CC:2263422.2271649} and malware analysis \cite{Bayer2009,Apel2009,Bailey:2007:ACA:1776434.1776449,Karim2005}, where sequences are longer than what most other techniques can handle ($10^4-10^8$ steps in length), and it can be difficult to extract more sophisticated features manually. However, using compression algorithms naively in NCD leads to difficulties with computational scalability and reduced accuracy/failure cases due to the fact that compression algorithms were not designed for similarity analysis \cite{Borbely2015,cebrian2005common}. For these reasons, some have looked at converting known clustering algorithms into explicit distance measures. For example, the Lempel Ziv Jaccard Distance (LZJD) \cite{raff_lzjd_2017} converts the LZMA algorithm into a true distance measure, and for malware classification tasks has superior accuracy and runtime compared to NCD\footnote{A Java \cite{raff_lzjd_digest} and Python  \cite{pylzjd-proc-scipy-2019} implementations of LZJD are available.}. %

\subsection{Extended Burrows Wheeler Transform} \label{sec:ebwt}

\setlength{\columnsep}{3.0pt}%
\setlength{\intextsep}{10.0pt}%
\begin{wraptable}[14]{r}{3.750cm}
\vspace{-13pt}
\caption{BWT }
\centering
\label{tbl:bwt_example}
\texttt{
\begin{tabular}{@{}ccc@{}}
\toprule
F  & Rotation    & L  \\ \midrule
\$ & \$easypeasy & y  \\
a  & asy\$easype & e  \\
a  & asypeasy\$e & e  \\
e  & easy\$easyp & p  \\
e  & easypeasy\$ & \$ \\
p  & peasy\$easy & y  \\
s  & sy\$easypea & a  \\
s  & sypeasy\$ea & a  \\
y  & y\$easypeas & s  \\
y  & ypeasy\$eas & s  \\ \bottomrule
\end{tabular}
}
\end{wraptable} 

The Burrows Wheeler Transform (BWT) \cite{Burrows1994} is a core component of the bzip compression algorithm, and has been widely used in information retrieval applications due to its ability to accelerate search queries\cite{Ferragina:2005:ICT:1082036.1082039}. The BWT takes an input string $u$ of length $n=|u|$, over an alphabet $\Sigma$, and produces a new string $u' = bwt(u)$. Through the use of an end-of-file (EoF) marker, the BWT is invertible, so $u$ can be recovered from $u'$ without loss.

BWT's utility in compression is best understood through example. Consider \autoref{tbl:bwt_example}, where the BWT transform of the string "easypeasy" is shown. BWT adds a special EoF marker "\$", and lexicographically sorts every single-character rotation of the string (observe column F, which is in sorted order). The BWT output is then the last column of each string, shown in column L. By computing the BWT version of the string, we can see how runs of the same character ("ee", "aa", "ss") have been created that previously did not exist. Simple run-length encoding can then be applied to produce a compressed version of the string. 

These sorted rotations can be computed in $\bigO(n)$ time, and provide a simple method of compression. To turn this into a distance measure, \textcite{Mantaci:2005:EBW:2134629.2134645} developed the Extended Burrows Wheeler Transform (EBWT). The EBWT works by defining the BWT over a pair of inputs $u$ and $v$, and computing the sorted order of both sequences. An example of this is shown in \autoref{tbl:ebwt_example}. 

\begin{table}[!h]
\caption{EBWT computation example} \label{tbl:ebwt_example}
\adjustbox{max width=\columnwidth}{%
\texttt{
\begin{tabular}{@{}lllcc@{}}
\toprule
\multicolumn{1}{c}{bwt($u$=bcaa)} & \multicolumn{1}{c}{bwt($v$=ccbab)} & \multicolumn{1}{c}{EBWT Merge} & Source  \\ \midrule
a a b c                         & a b c c b                        & a a b c                        & $u$      \\
a b c a                         & b a b c c                        & a b c a                        & $u$      \\
b c a a                         & b c c b a                        & a b c c b                      & $v$      \\
c a a b                         & c b a b c                        & b a b c c                      & $v$      \\
\multicolumn{1}{c}{---}         & c c b a b                        & b c a a                        & $u$      \\
\multicolumn{1}{c}{---}         &\multicolumn{1}{c}{---}           & b c c b a                      & $v$      \\
\multicolumn{1}{c}{---}         &\multicolumn{1}{c}{---}           & c a a b                        & $u$      \\
\multicolumn{1}{c}{---}         &\multicolumn{1}{c}{---}           & c b a b c                      & $v$      \\
\multicolumn{1}{c}{---}         &\multicolumn{1}{c}{---}           & c c b a b                      & $v$      \\ \bottomrule
\end{tabular}
}
}
\end{table}

The distance between the two sequences $u$ and $v$
is defined by \eqref{eq:ebwt}, where $\text{rep}(i)$ returns how many times the $i$'th source occurred in a row, and that only $t$ source transitions occurred.
\begin{equation} \label{eq:ebwt}
ebwt(u,v) = \sum_{i=1}^t \max(\text{rep}(i)-1,0)
\end{equation}
 Again, this concept is easier understood through example. Considering \autoref{tbl:ebwt_example}, we can see that the source string sequence is $uuvvuvuvv$. If we group this by transitions, we get $u^2 v^2 u v u v^2$. Thus the distance $ebwt(u,v) = 1 + 1 + 0 + 0 +  0 + 1 = 3$. 
 
\textcite{Mantaci:2005:EBW:2134629.2134645} developed the EBWT for applications in bioinformatics, and developed theory to show a number of situations under which the EBWT will perform well or have desirable properties. However, it is still expensive to compute, requiring $\bigO(|u|+|v|)$ time for every distance computation. This makes EBWT less attractive as bioinformatic sequences become longer, and reduces its utility in other domains in which compression distances have found use, such as malware classification. While it was known that EBWT did not satisfy the triangle inequality, preventing it from being a true distance metric, previously unreported theoretical issues also exist. We will discuss these issues in \autoref{sec:theoretical_results}. 

\section{Burrows Wheeler Markov Distance} \label{sec:bwmd}

Inspired by the prior work we have just discussed, we now develop a new distance measure based on the Burrows Wheeler Transform. We will refer to our method as the Burrows Wheeler Markov Distance (BWMD), and it is simple to implement.  To begin, consider again \autoref{tbl:bwt_example}, where BWT("easypeasy") is shown. The BWT's effectiveness as a compression algorithm comes explicitly from its ability to re-order the content such that repetitions are reduced to a first-order occurrence. This is why run-length encoding  
is effective. Because first-order compression is independently effective after the BWT, we do not need to consider more complex interactions over the extent of a file.

We also do not care about the invertibility of any transform. Our goal is to define a new feature space where we can perform effective machine learning and information retrieval. So we seek to build a small statistical summary of the data, rather than 
build an object from which we can recreate the original data. %
By measuring the similarity of these small summaries, we measure the similarity of the underlying sequences. Given that first-order compression is effective with the BWT, we chose to select a first-order statistical model. In particular, we can use a Markov model of the probability of observing token $u'_i$ given the previous token $u'_{i-1}$. The transition matrix $T \in \mathbb{R}^{|\Sigma| \times |\Sigma|}$ can then be used as a statistical summary of the entire sequence BWT($u$). 

This is all that is needed to describe BWMD, and a succinct description is given below. Each step takes $\bigO(n)$ time for an input sequence of $n$ bytes. $\mathbbm{1}[z]$ is the indicator function, which returns $1$ if $z$ is true, and $0$ otherwise, and $\alpha$, $\beta$ are the rows and columns of the transition matrix. 
\begin{enumerate}
\item For each sequence $u$ in a corpus of size $N$ 
\item Compute $u'= \text{BWT}(u)$
\item Estimate flattened Markov transition vector $$x[\alpha+\beta\cdot |\Sigma|] = \frac{1}{|u'|-1}\sum_{i=2}^{|u'|}  \mathbbm{1}\left[u'_i = \alpha \land u'_{i-1} = \beta\right]$$
\item Normalize $x$ such that $x[i] = \sqrt{x[i]}/\sqrt{2}$. 
\end{enumerate}

After step (4) in the above process, we obtain from the input $u$ a single feature vector $x$ which we might use, in place of $u$, in machine learning or information retrieval algorithms.
Regardless of the length of the input sequence $u$, the size of the vector $x$ will depend only on the alphabet size $|\Sigma|$. When working on raw bytes, this would be a $256^2$ feature vector. While such a vector takes up to 256 KB per sequence $u$, the individual input data objects we consider in this work range from 1 MB in size up to 400 MB. This makes the BWMD description quite compact by comparison. When $u$ is shorter in length, the vector $x$ can be stored in a sparse form, making the memory cost $\bigO(\min(|u|, |\Sigma|^2))$

The normalization  in step (4) is also chosen intentionally. The Hellinger Distance %
is a metric over probability distributions. For the case of two discrete probability distributions $P = (p_1, \ldots, p_k)$ and $Q = (q_1, \ldots, q_k) $, the Hellinger Distance is defined as follows: %

\begin{equation} \label{eq:hellinger}
    H ( P , Q ) = \frac { 1 } { \sqrt { 2 } } \sqrt { \sum _ { i = 1 } ^ { k } \left( \sqrt { p _ { i } } - \sqrt { q _ { i } } \right) ^ { 2 } }
\end{equation}

Due to the form of \eqref{eq:hellinger}, the Hellinger distance corresponds to the Euclidean distance between two transformed and scaled versions of $P$ and $Q$. 
By using the square root of the coefficients in step (4), divided by $\sqrt{2}$, we have a feature vector $x$ that has been placed into a space where the Euclidean distance can be computed as usual. The results are then equivalent to the Hellinger distance between Markov transition vectors, giving us a statistically sound interpretation for BWMD. That the BWMD is a valid distance metric also follows immediately from the use of the Euclidean distance, which is well known to be a valid metric, a property that the EBWT lacks. 

We explicitly use the Hellinger distance over other alternatives, like KL-Divergence, because it corresponds exactly to the Euclidean distance after transformation. This makes is a valid distance metric 
for which we can make mathematical statements about its behavior with little work, and prove it does not have the same shortcomings as prior methods in this space. 
In addition, most other clustering and fast retrieval algorithms are built around the Euclidean distance, making our method compatible with a maximal number of complimentary techniques. This is not the case for any prior compression based measure. The ability to use algorithms like k-means, where other alternatives cannot, provides BWMD with advantages in terms of clustering accuracy, as well as computational efficiency to handle big data.

Note that because of the BWT, our approach is \textit{not} equivalent to 2-grams. Our comparison with LZJD, which can be interpreted as adaptive variable-length gram\cite{raff_shwel}, and BitShred, which uses 16-grams, will show that our method is meaningfully more effective than simple n-gram approaches.

\section{Theoretical Results} \label{sec:theoretical_results}

We begin by developing a stronger theoretical understanding of our new method, as well as the prior approach EBWT. Prior works have looked at a number of properties of the EBWT\cite{Mantaci2007,Mantaci2008,Mantaci:2005:EBW:2134629.2134645}, and describe situations in which EBWT will behave as a metric for a subset of possible inputs, and that it is invertible like the standard BWT. However, our interest in BWT is as a general purpose similarity measure for information retrieval and machine learning applications. We begin by showing three undesirable properties of the EBWT that reduce our confidence in its use for such applications. Then we will investigate the nature of our new BWMD in these same cases where EBWT might fail. 

\subsection{EBWT Shortcomings }

First we show a simple property that is a direct result of the EBWT measuring distance as a function of repeated source sequences. When we have two strings $u$ and $v$ in any alphabet $\Sigma$, it is necessarily the case that the distance is bounded below by the difference in sequence lengths $|u|$ and $|v|$. If $u$ and $v$ differ significantly, we are unlikely to be able to make meaningful similarity comparisons.  

\begin{theorem}\label{thm:ebwt_1}
The distance $ebwt(u, v) \geq \left||u|-|v|\right|-1 $. 
\end{theorem}
\begin{proof}
Consider any two strings $u$ and $v$. The minimum distance involves the maximum number of transitions between string sources. If $|u| < |v|$, that means there can be at most $2\cdot |u|$ transitions, going back and forth between $u$ and $v$ on the merging. That necessitates $|v|-|u|$ repetitions at the end. Given the definition of $ebwt$ in \eqref{eq:ebwt}, that means the minimum distance between $u$ and $v$ must be $|v|-|u|-1$. 
\end{proof}

The above result leads us to suspect that EBWT will not be useful when the sequences being compared are of varying lengths. The greater the difference in sequence length, the more troubling this issue might become.  Given the insight from \autoref{thm:ebwt_1}, we move on to a more serious departure from our intuition of how a distance measure should behave. In particular, if $u$ is a subset of $v$, we should expect the distance between $u$ and $v$ to be small. Instead, it is possible for EBWT to return the maximal distance under this scenario. 

\begin{theorem} \label{thm:ebwt_bad_max}
It is possible for $ebwt(u,v) = |v|+|u|-2$, the maximum possible distance, even if $u \subset v$. 
\end{theorem}
\begin{proof}
Consider the string $u=a^{n_1}$ and $v=a^{n_2}$, such that $n_2 > n_1$, $|u| = n_1$ and $|v|=n_2$. Because of the topographical sorting, all rotations of $u$ and $v$ will have the same characters, and so the sorting will only resolve once the max substring length is reached. Since all rotations in $u$ are of length $n_1$, which is shorter than $n_2$, a sorting will place all rotations of $u$ before any rotations of $v$. This results in a transition pattern of $u^{n_1} v^{n_2}$, and thus $ebwt(u,v)=n_1-1 + n_2-1 = |u|+|v|-2$. 
\end{proof}

\autoref{thm:ebwt_bad_max} defies our expectations in the case of similar inputs. If we use the behavior of the NID from \eqref{eq:nid}, we would expect the distance in this scenario to be small. Consider the proof's example with $u=a^{n_1}$ and $v=a^{n_2}$: we would expect NID($u$,$v$) $\leq  \log(n_2/n_1)$. This is because $v$ could be encoded as the sequence $u$ repeated $n_2/n_1$ times, which in the worst case, can be represented in a number of bits logarithmic in the value being encoded (i.e., the nature of any big-integer representation is that the maximum value that can be represented grows exponentially with a doubling of the bits). 

We now show that EBWT likewise surprises us in the case of dissimilar inputs,  where we can have $u$ and $v$ with no overlap in content, but EBWT identifies as having maximal similarity. 

\begin{theorem} \label{thm:ebwt_bad_0}
It is possible for $ebwt(u,v) = 0$, even if there exists no shared characters between $u$ and $v$. More formally, there exists $u,v$ and an alphabet $|\Sigma| \geq |u|+|v|$ such that for every $x \in u$ and $z \in v$, $x \notin v$ $\land$ $z \notin u$ yet $ebwt(u,v) = 0$. 
\end{theorem}
\begin{proof}
Let $a \Vert b$ denote the concatenation of $a$ and $b$, and $\big\Vert_{i=1}^N i $ the sequential concatenation of $1, 2, \ldots, N$. Without loss of generality, define the string $u=\big\Vert_{i=1}^{n_0} 2\cdot i$ and $v=\big\Vert_{i=1}^{n_0} 2\cdot i+1$. Thus $|u|=|v|$, but in the lexicographical sorting $u$ and $v$ will alternate between rotations $u[0]=2$, $v[0]=3$, $u[1]=4$, $v[1]=5$, $\ldots$, $u[n_0]=2\cdot n_0$, $v[n_0]=2\cdot n_0+1$. Thus $ebwt(u,v)$ will always contain transitions with no source repetition, hence $ebwt(u,v) = 0$. 
\end{proof}

The reader may note that the construction of \autoref{thm:ebwt_bad_0} would allow one to argue that the scenario should return a small distance under the ideal Kolmogorov distance with NID. This could be argued because the construction of $u$ and $v$ allows $v$ to be represented as $v = u +1$. However, the same scenario can occur with randomized strings where the alphabet does not increment with any simple pattern  and is filled with random tokens, so long as there is no overlap in the tokens and the tokens "balance out" once sorted (i.e., there is no $f(\cdot) \text{ s.t. } v[i] = f(u[i])$ yet $v[i] > u[i] \land u[i+1] > v[i] \forall i$). This requires further expanding the alphabet size $|\Sigma|$, and as such makes \autoref{thm:ebwt_bad_0} the least practical of our concerns. However, we find it enlightening as a theoretical shortcoming which we would prefer to avoid.

While these issues give us cause for concern, EBWT has found use in practice. We will show in \autoref{sec:dna} that our concern for EBWT's utility is more justified when sequences have varying length. 

\subsection{Behaviors of BWMD}

To delve into BWMD's behavior, we will begin by analyzing the same scenarios used to show our theoretical concerns with the EBWT in the preceding section, as well as compare the behavior of BWMD to that of the LZJD algorithm. 

\begin{corollary} \label{thm:bwmd_repitition}
Given $u=a^{n_1}$ and $v=a^{n_2}$, such that $n_2 > n_1$, then $bwmd(u,v)=0$. 
\end{corollary}
\begin{proof}
Under the construction of the embedding of $u$,  $x_u \in \mathbb{R}^{|\Sigma|}$, $\|x\|_0 = 1$ since there will be only one transition pattern of $a-> a$. As such, the value at that index $i$ will be $x_u[i] = \sqrt{1}/\sqrt{2}= 1/\sqrt{2}$. Since the embedding $x_v$ will have the same construction, and thus, $x_u[i] - x_v[i] = 0$, and $\forall j \neq i, x_u[j] = x_v[j] = 0$. Therefore, the distance between $u$ and $v$ will be zero. 
\end{proof}

The above also does not conform to expectation with respect to the NID, because we are ignoring the length of the inputs in our computation of distances. The NID($u$,$v$) would be greater than zero in this scenario and is necessitated by storing the difference in repetition lengths $n_2$ and $n_1$.  This tells us BWMD will be less sensitive to differences in sequence length, which may be desirable or not, depending on the application. 

The behavior of LZJD in this scenario was used to prove its sensitivity  to potential repetition of the input. It was shown that $\text{LZJD}(a^{n_1},a^{n_2}) = 1-\frac{\sqrt{8 n_1+1}-1}{\sqrt{8 n_2+1}-1}$ as a lower bound for a similar scenario \cite{raff_lzjd_2017}. This distance grows at a rate considerably faster than logarithmic, but is also better than the EBWT distance in this case. We conclude that both LZJD and BWMD have better behavior, but BWMD will lower-bound the NID and LZJD will upper-bound the NID. 

In a similar manner as \autoref{thm:bwmd_repitition} was shown, the same construction can be applied to \autoref{thm:ebwt_bad_0}'s issue for BWMD. 

\begin{corollary}
For all $u, v$ such that for every $z \in v$, $x \notin v$ $\land$ $z \notin u$, then $bwmd(u,v)=1$, the maximum possible distance. 
\end{corollary}
The derivation follows from the fact that $\sum_{i=1}^k \sqrt{p_i}^2 = \sum_{i=1}^k p_i = 1$. This means the distance computation will reduce to $1/\sqrt{2}\sqrt{2\cdot1} = 1$. Therefore, the distance when the embeddings $x_u$ and $x_v$ have no intersection is maximized. In this case, BWMD aligns well with the behavior we would expect from the NID. Likewise it is easy to see that LZJD will return its maximum distance of 1 in this scenario as well. LZJD measures the set intersection, so when the sets have no intersection, then maximal distance is achieved.

\section{Genomic Clustering} \label{sec:dna}

We begin by showing that our new BWMD has similar utility as the original EBWT distance for genomic phylogeny from DNA sequences. This was the original proposed use of the EBWT measure, where they evaluated the Single-Link Clustering results on mitochondrial DNA (mtDNA) \cite{Mantaci:2005:EBW:2134629.2134645}. Such data can be obtained using the NIH GeneBank
, which we have used to create a similar corpus of DNA sequences to compare the relative pros and cons of BWMD and EBWT. We will use both mtDNA as has been done in prior work, but also a more challenging case with chromosomal genomic scaffold DNA sequences. We will produce dendrograms for each tasks with Single-Link Clustering (SLINK).

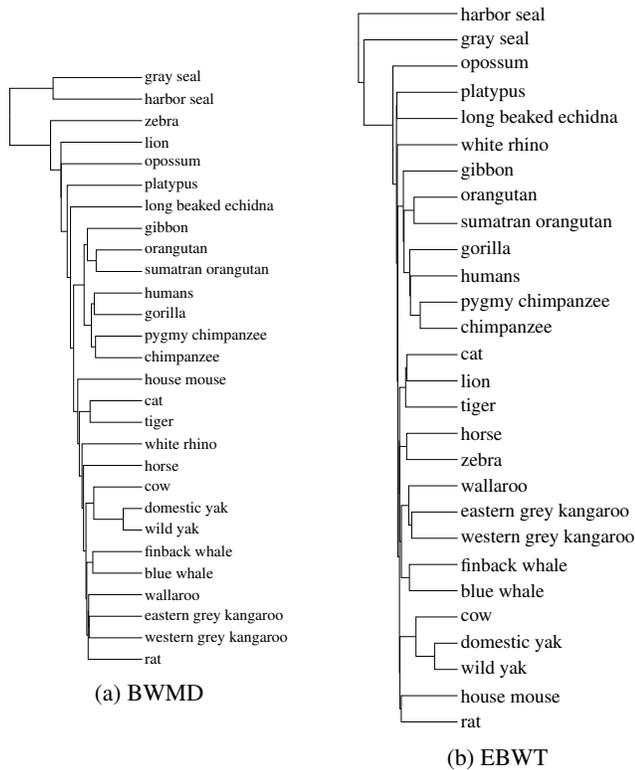
\begin{figure}[!htb]
\centering
\begin{subfigure}{0.45\columnwidth}
\centering
\adjustbox{max width=\columnwidth}{%
    \begin{turn}{90}
        \begin{tikzpicture}[sloped]
\tikzset{anchor=west}
\node[rotate=-90] (whiterhino) at (10.0,0.0) {\Huge white rhino};
\node[rotate=-90] (wallaroo) at (3.0,0.0) {\Huge wallaroo};
\node[rotate=-90] (cow) at (8.0,0.0) {\Huge cow};
\node[rotate=-90] (bluewhale) at (4.0,0.0) {\Huge blue whale};
\node[rotate=-90] (housemouse) at (13.0,0.0) {\Huge house mouse};
\node[rotate=-90] (finbackwhale) at (5.0,0.0) {\Huge finback whale};
\node[rotate=-90] (cat) at (12.0,0.0) {\Huge cat};
\node[rotate=-90] (domesticyak) at (7.0,0.0) {\Huge domestic yak};
\node[rotate=-90] (sumatranorangutan) at (18.0,0.0) {\Huge sumatran orangutan};
\node[rotate=-90] (chimpanzee) at (14.0,0.0) {\Huge chimpanzee};
\node[rotate=-90] (zebra) at (25.0,0.0) {\Huge zebra};
\node[rotate=-90] (gorilla) at (16.0,0.0) {\Huge gorilla};
\node[rotate=-90] (platypus) at (22.0,0.0) {\Huge platypus};
\node[rotate=-90] (opossum) at (23.0,0.0) {\Huge opossum};
\node[rotate=-90] (rat) at (0.0,0.0) {\Huge rat};
\node[rotate=-90] (grayseal) at (27.0,0.0) {\Huge gray seal};
\node[rotate=-90] (orangutan) at (19.0,0.0) {\Huge orangutan};
\node[rotate=-90] (harborseal) at (26.0,0.0) {\Huge harbor seal};
\node[rotate=-90] (westerngreykangaroo) at (1.0,0.0) {\Huge western grey kangaroo};
\node[rotate=-90] (tiger) at (11.0,0.0) {\Huge tiger};
\node[rotate=-90] (lion) at (24.0,0.0) {\Huge lion};
\node[rotate=-90] (horse) at (9.0,0.0) {\Huge horse};
\node[rotate=-90] (easterngreykangaroo) at (2.0,0.0) {\Huge eastern grey kangaroo};
\node[rotate=-90] (longbeakedechidna) at (21.0,0.0) {\Huge long beaked echidna};
\node[rotate=-90] (gibbon) at (20.0,0.0) {\Huge gibbon};
\node[rotate=-90] (pygmychimpanzee) at (15.0,0.0) {\Huge pygmy chimpanzee};
\node[rotate=-90] (wildyak) at (6.0,0.0) {\Huge wild yak};
\node[rotate=-90] (humans) at (17.0,0.0) {\Huge humans};
\node[rotate=-90] (wildyakdomesticyak) at (6.5,1.0) {\Huge };
\draw  (wildyak) |- (wildyakdomesticyak.center);
\draw  (domesticyak) |- (wildyakdomesticyak.center);
\node[rotate=-90] (sumatranorangutanorangutan) at (18.5,2.2537307751359874) {\Huge };
\draw  (sumatranorangutan) |- (sumatranorangutanorangutan.center);
\draw  (orangutan) |- (sumatranorangutanorangutan.center);
\node[rotate=-90] (chimpanzeepygmychimpanzee) at (14.5,2.2845593384579734) {\Huge };
\draw  (chimpanzee) |- (chimpanzeepygmychimpanzee.center);
\draw  (pygmychimpanzee) |- (chimpanzeepygmychimpanzee.center);
\node[rotate=-90] (gorillahumans) at (16.5,2.343273638772904) {\Huge };
\draw  (gorilla) |- (gorillahumans.center);
\draw  (humans) |- (gorillahumans.center);
\node[rotate=-90] (wildyakdomesticyakcow) at (7.25,2.365584712159716) {\Huge };
\draw  (wildyakdomesticyak.center) |- (wildyakdomesticyakcow.center);
\draw  (cow) |- (wildyakdomesticyakcow.center);
\node[rotate=-90] (bluewhalefinbackwhale) at (4.5,2.4104475630389035) {\Huge };
\draw  (bluewhale) |- (bluewhalefinbackwhale.center);
\draw  (finbackwhale) |- (bluewhalefinbackwhale.center);
\node[rotate=-90] (chimpanzeepygmychimpanzeegorillahumans) at (15.5,2.4770548395173497) {\Huge };
\draw  (chimpanzeepygmychimpanzee.center) |- (chimpanzeepygmychimpanzeegorillahumans.center);
\draw  (gorillahumans.center) |- (chimpanzeepygmychimpanzeegorillahumans.center);
\node[rotate=-90] (tigercat) at (11.5,2.5364238523210902) {\Huge };
\draw  (tiger) |- (tigercat.center);
\draw  (cat) |- (tigercat.center);
\node[rotate=-90] (westerngreykangarooeasterngreykangaroo) at (1.5,2.57930349529655) {\Huge };
\draw  (westerngreykangaroo) |- (westerngreykangarooeasterngreykangaroo.center);
\draw  (easterngreykangaroo) |- (westerngreykangarooeasterngreykangaroo.center);
\node[rotate=-90] (westerngreykangarooeasterngreykangaroowallaroo) at (2.25,2.603263401455906) {\Huge };
\draw  (westerngreykangarooeasterngreykangaroo.center) |- (westerngreykangarooeasterngreykangaroowallaroo.center);
\draw  (wallaroo) |- (westerngreykangarooeasterngreykangaroowallaroo.center);
\node[rotate=-90] (ratwesterngreykangarooeasterngreykangaroowallaroo) at (1.125,2.62641177515117) {\Huge };
\draw  (rat) |- (ratwesterngreykangarooeasterngreykangaroowallaroo.center);
\draw  (westerngreykangarooeasterngreykangaroowallaroo.center) |- (ratwesterngreykangarooeasterngreykangaroowallaroo.center);
\node[rotate=-90] (sumatranorangutanorangutangibbon) at (19.25,2.660915033573577) {\Huge };
\draw  (sumatranorangutanorangutan.center) |- (sumatranorangutanorangutangibbon.center);
\draw  (gibbon) |- (sumatranorangutanorangutangibbon.center);
\node[rotate=-90] (ratwesterngreykangarooeasterngreykangaroowallaroobluewhalefinbackwhale) at (2.8125,2.690735675224926) {\Huge };
\draw  (ratwesterngreykangarooeasterngreykangaroowallaroo.center) |- (ratwesterngreykangarooeasterngreykangaroowallaroobluewhalefinbackwhale.center);
\draw  (bluewhalefinbackwhale.center) |- (ratwesterngreykangarooeasterngreykangaroowallaroobluewhalefinbackwhale.center);
\node[rotate=-90] (ratwesterngreykangarooeasterngreykangaroowallaroobluewhalefinbackwhalewildyakdomesticyakcow) at (5.03125,2.728281003579813) {\Huge };
\draw  (ratwesterngreykangarooeasterngreykangaroowallaroobluewhalefinbackwhale.center) |- (ratwesterngreykangarooeasterngreykangaroowallaroobluewhalefinbackwhalewildyakdomesticyakcow.center);
\draw  (wildyakdomesticyakcow.center) |- (ratwesterngreykangarooeasterngreykangaroowallaroobluewhalefinbackwhalewildyakdomesticyakcow.center);
\node[rotate=-90] (chimpanzeepygmychimpanzeegorillahumanssumatranorangutanorangutangibbon) at (17.375,2.814449218073787) {\Huge };
\draw  (chimpanzeepygmychimpanzeegorillahumans.center) |- (chimpanzeepygmychimpanzeegorillahumanssumatranorangutanorangutangibbon.center);
\draw  (sumatranorangutanorangutangibbon.center) |- (chimpanzeepygmychimpanzeegorillahumanssumatranorangutanorangutangibbon.center);
\node[rotate=-90] (ratwesterngreykangarooeasterngreykangaroowallaroobluewhalefinbackwhalewildyakdomesticyakcowhorse) at (7.015625,2.8412551225740987) {\Huge };
\draw  (ratwesterngreykangarooeasterngreykangaroowallaroobluewhalefinbackwhalewildyakdomesticyakcow.center) |- (ratwesterngreykangarooeasterngreykangaroowallaroobluewhalefinbackwhalewildyakdomesticyakcowhorse.center);
\draw  (horse) |- (ratwesterngreykangarooeasterngreykangaroowallaroobluewhalefinbackwhalewildyakdomesticyakcowhorse.center);
\node[rotate=-90] (ratwesterngreykangarooeasterngreykangaroowallaroobluewhalefinbackwhalewildyakdomesticyakcowhorsewhiterhino) at (8.5078125,2.910637313603119) {\Huge };
\draw  (ratwesterngreykangarooeasterngreykangaroowallaroobluewhalefinbackwhalewildyakdomesticyakcowhorse.center) |- (ratwesterngreykangarooeasterngreykangaroowallaroobluewhalefinbackwhalewildyakdomesticyakcowhorsewhiterhino.center);
\draw  (whiterhino) |- (ratwesterngreykangarooeasterngreykangaroowallaroobluewhalefinbackwhalewildyakdomesticyakcowhorsewhiterhino.center);
\node[rotate=-90] (ratwesterngreykangarooeasterngreykangaroowallaroobluewhalefinbackwhalewildyakdomesticyakcowhorsewhiterhinotigercat) at (10.00390625,3.036514986290469) {\Huge };
\draw  (ratwesterngreykangarooeasterngreykangaroowallaroobluewhalefinbackwhalewildyakdomesticyakcowhorsewhiterhino.center) |- (ratwesterngreykangarooeasterngreykangaroowallaroobluewhalefinbackwhalewildyakdomesticyakcowhorsewhiterhinotigercat.center);
\draw  (tigercat.center) |- (ratwesterngreykangarooeasterngreykangaroowallaroobluewhalefinbackwhalewildyakdomesticyakcowhorsewhiterhinotigercat.center);
\node[rotate=-90] (ratwesterngreykangarooeasterngreykangaroowallaroobluewhalefinbackwhalewildyakdomesticyakcowhorsewhiterhinotigercathousemouse) at (11.501953125,3.1102788367747274) {\Huge };
\draw  (ratwesterngreykangarooeasterngreykangaroowallaroobluewhalefinbackwhalewildyakdomesticyakcowhorsewhiterhinotigercat.center) |- (ratwesterngreykangarooeasterngreykangaroowallaroobluewhalefinbackwhalewildyakdomesticyakcowhorsewhiterhinotigercathousemouse.center);
\draw  (housemouse) |- (ratwesterngreykangarooeasterngreykangaroowallaroobluewhalefinbackwhalewildyakdomesticyakcowhorsewhiterhinotigercathousemouse.center);
\node[rotate=-90] (ratwesterngreykangarooeasterngreykangaroowallaroobluewhalefinbackwhalewildyakdomesticyakcowhorsewhiterhinotigercathousemousechimpanzeepygmychimpanzeegorillahumanssumatranorangutanorangutangibbon) at (14.4384765625,3.3005010895852434) {\Huge };
\draw  (ratwesterngreykangarooeasterngreykangaroowallaroobluewhalefinbackwhalewildyakdomesticyakcowhorsewhiterhinotigercathousemouse.center) |- (ratwesterngreykangarooeasterngreykangaroowallaroobluewhalefinbackwhalewildyakdomesticyakcowhorsewhiterhinotigercathousemousechimpanzeepygmychimpanzeegorillahumanssumatranorangutanorangutangibbon.center);
\draw  (chimpanzeepygmychimpanzeegorillahumanssumatranorangutanorangutangibbon.center) |- (ratwesterngreykangarooeasterngreykangaroowallaroobluewhalefinbackwhalewildyakdomesticyakcowhorsewhiterhinotigercathousemousechimpanzeepygmychimpanzeegorillahumanssumatranorangutanorangutangibbon.center);
\node[rotate=-90] (ratwesterngreykangarooeasterngreykangaroowallaroobluewhalefinbackwhalewildyakdomesticyakcowhorsewhiterhinotigercathousemousechimpanzeepygmychimpanzeegorillahumanssumatranorangutanorangutangibbonlongbeakedechidna) at (17.71923828125,3.4453297674172885) {\Huge };
\draw  (ratwesterngreykangarooeasterngreykangaroowallaroobluewhalefinbackwhalewildyakdomesticyakcowhorsewhiterhinotigercathousemousechimpanzeepygmychimpanzeegorillahumanssumatranorangutanorangutangibbon.center) |- (ratwesterngreykangarooeasterngreykangaroowallaroobluewhalefinbackwhalewildyakdomesticyakcowhorsewhiterhinotigercathousemousechimpanzeepygmychimpanzeegorillahumanssumatranorangutanorangutangibbonlongbeakedechidna.center);
\draw  (longbeakedechidna) |- (ratwesterngreykangarooeasterngreykangaroowallaroobluewhalefinbackwhalewildyakdomesticyakcowhorsewhiterhinotigercathousemousechimpanzeepygmychimpanzeegorillahumanssumatranorangutanorangutangibbonlongbeakedechidna.center);
\node[rotate=-90] (ratwesterngreykangarooeasterngreykangaroowallaroobluewhalefinbackwhalewildyakdomesticyakcowhorsewhiterhinotigercathousemousechimpanzeepygmychimpanzeegorillahumanssumatranorangutanorangutangibbonlongbeakedechidnaplatypus) at (19.859619140625,3.597034302332711) {\Huge };
\draw  (ratwesterngreykangarooeasterngreykangaroowallaroobluewhalefinbackwhalewildyakdomesticyakcowhorsewhiterhinotigercathousemousechimpanzeepygmychimpanzeegorillahumanssumatranorangutanorangutangibbonlongbeakedechidna.center) |- (ratwesterngreykangarooeasterngreykangaroowallaroobluewhalefinbackwhalewildyakdomesticyakcowhorsewhiterhinotigercathousemousechimpanzeepygmychimpanzeegorillahumanssumatranorangutanorangutangibbonlongbeakedechidnaplatypus.center);
\draw  (platypus) |- (ratwesterngreykangarooeasterngreykangaroowallaroobluewhalefinbackwhalewildyakdomesticyakcowhorsewhiterhinotigercathousemousechimpanzeepygmychimpanzeegorillahumanssumatranorangutanorangutangibbonlongbeakedechidnaplatypus.center);
\node[rotate=-90] (ratwesterngreykangarooeasterngreykangaroowallaroobluewhalefinbackwhalewildyakdomesticyakcowhorsewhiterhinotigercathousemousechimpanzeepygmychimpanzeegorillahumanssumatranorangutanorangutangibbonlongbeakedechidnaplatypusopossum) at (21.4298095703125,3.876861306149897) {\Huge };
\draw  (ratwesterngreykangarooeasterngreykangaroowallaroobluewhalefinbackwhalewildyakdomesticyakcowhorsewhiterhinotigercathousemousechimpanzeepygmychimpanzeegorillahumanssumatranorangutanorangutangibbonlongbeakedechidnaplatypus.center) |- (ratwesterngreykangarooeasterngreykangaroowallaroobluewhalefinbackwhalewildyakdomesticyakcowhorsewhiterhinotigercathousemousechimpanzeepygmychimpanzeegorillahumanssumatranorangutanorangutangibbonlongbeakedechidnaplatypusopossum.center);
\draw  (opossum) |- (ratwesterngreykangarooeasterngreykangaroowallaroobluewhalefinbackwhalewildyakdomesticyakcowhorsewhiterhinotigercathousemousechimpanzeepygmychimpanzeegorillahumanssumatranorangutanorangutangibbonlongbeakedechidnaplatypusopossum.center);
\node[rotate=-90] (ratwesterngreykangarooeasterngreykangaroowallaroobluewhalefinbackwhalewildyakdomesticyakcowhorsewhiterhinotigercathousemousechimpanzeepygmychimpanzeegorillahumanssumatranorangutanorangutangibbonlongbeakedechidnaplatypusopossumlion) at (22.71490478515625,3.8873787400453894) {\Huge };
\draw  (ratwesterngreykangarooeasterngreykangaroowallaroobluewhalefinbackwhalewildyakdomesticyakcowhorsewhiterhinotigercathousemousechimpanzeepygmychimpanzeegorillahumanssumatranorangutanorangutangibbonlongbeakedechidnaplatypusopossum.center) |- (ratwesterngreykangarooeasterngreykangaroowallaroobluewhalefinbackwhalewildyakdomesticyakcowhorsewhiterhinotigercathousemousechimpanzeepygmychimpanzeegorillahumanssumatranorangutanorangutangibbonlongbeakedechidnaplatypusopossumlion.center);
\draw  (lion) |- (ratwesterngreykangarooeasterngreykangaroowallaroobluewhalefinbackwhalewildyakdomesticyakcowhorsewhiterhinotigercathousemousechimpanzeepygmychimpanzeegorillahumanssumatranorangutanorangutangibbonlongbeakedechidnaplatypusopossumlion.center);
\node[rotate=-90] (harborsealgrayseal) at (26.5,4.251693715126765) {\Huge };
\draw  (harborseal) |- (harborsealgrayseal.center);
\draw  (grayseal) |- (harborsealgrayseal.center);
\node[rotate=-90] (ratwesterngreykangarooeasterngreykangaroowallaroobluewhalefinbackwhalewildyakdomesticyakcowhorsewhiterhinotigercathousemousechimpanzeepygmychimpanzeegorillahumanssumatranorangutanorangutangibbonlongbeakedechidnaplatypusopossumlionzebra) at (23.857452392578125,4.3702614199287195) {\Huge };
\draw  (ratwesterngreykangarooeasterngreykangaroowallaroobluewhalefinbackwhalewildyakdomesticyakcowhorsewhiterhinotigercathousemousechimpanzeepygmychimpanzeegorillahumanssumatranorangutanorangutangibbonlongbeakedechidnaplatypusopossumlion.center) |- (ratwesterngreykangarooeasterngreykangaroowallaroobluewhalefinbackwhalewildyakdomesticyakcowhorsewhiterhinotigercathousemousechimpanzeepygmychimpanzeegorillahumanssumatranorangutanorangutangibbonlongbeakedechidnaplatypusopossumlionzebra.center);
\draw  (zebra) |- (ratwesterngreykangarooeasterngreykangaroowallaroobluewhalefinbackwhalewildyakdomesticyakcowhorsewhiterhinotigercathousemousechimpanzeepygmychimpanzeegorillahumanssumatranorangutanorangutangibbonlongbeakedechidnaplatypusopossumlionzebra.center);
\node[rotate=-90] (ratwesterngreykangarooeasterngreykangaroowallaroobluewhalefinbackwhalewildyakdomesticyakcowhorsewhiterhinotigercathousemousechimpanzeepygmychimpanzeegorillahumanssumatranorangutanorangutangibbonlongbeakedechidnaplatypusopossumlionzebraharborsealgrayseal) at (25.178726196289062,6.259904315031815) {\Huge };
\draw  (ratwesterngreykangarooeasterngreykangaroowallaroobluewhalefinbackwhalewildyakdomesticyakcowhorsewhiterhinotigercathousemousechimpanzeepygmychimpanzeegorillahumanssumatranorangutanorangutangibbonlongbeakedechidnaplatypusopossumlionzebra.center) |- (ratwesterngreykangarooeasterngreykangaroowallaroobluewhalefinbackwhalewildyakdomesticyakcowhorsewhiterhinotigercathousemousechimpanzeepygmychimpanzeegorillahumanssumatranorangutanorangutangibbonlongbeakedechidnaplatypusopossumlionzebraharborsealgrayseal.center);
\draw  (harborsealgrayseal.center) |- (ratwesterngreykangarooeasterngreykangaroowallaroobluewhalefinbackwhalewildyakdomesticyakcowhorsewhiterhinotigercathousemousechimpanzeepygmychimpanzeegorillahumanssumatranorangutanorangutangibbonlongbeakedechidnaplatypusopossumlionzebraharborsealgrayseal.center);
\end{tikzpicture}
    \end{turn}
}
\caption{BWMD} \label{fig:mtdna_bwmd}
\end{subfigure}
\hspace*{\fill} %
\begin{subfigure}{0.45\columnwidth}
\centering
\adjustbox{max width=\columnwidth}{%
    \begin{turn}{90}
        \begin{tikzpicture}[sloped]
\tikzset{anchor=west}
\node[rotate=-90] (whiterhino) at (22.0,0.0) {\Huge white rhino};
\node[rotate=-90] (wallaroo) at (9.0,0.0) {\Huge wallaroo};
\node[rotate=-90] (cow) at (4.0,0.0) {\Huge cow};
\node[rotate=-90] (bluewhale) at (5.0,0.0) {\Huge blue whale};
\node[rotate=-90] (housemouse) at (1.0,0.0) {\Huge house mouse};
\node[rotate=-90] (finbackwhale) at (6.0,0.0) {\Huge finback whale};
\node[rotate=-90] (cat) at (14.0,0.0) {\Huge cat};
\node[rotate=-90] (domesticyak) at (3.0,0.0) {\Huge domestic yak};
\node[rotate=-90] (sumatranorangutan) at (19.0,0.0) {\Huge sumatran orangutan};
\node[rotate=-90] (chimpanzee) at (15.0,0.0) {\Huge chimpanzee};
\node[rotate=-90] (zebra) at (10.0,0.0) {\Huge zebra};
\node[rotate=-90] (gorilla) at (18.0,0.0) {\Huge gorilla};
\node[rotate=-90] (platypus) at (24.0,0.0) {\Huge platypus};
\node[rotate=-90] (opossum) at (25.0,0.0) {\Huge opossum};
\node[rotate=-90] (rat) at (0.0,0.0) {\Huge rat};
\node[rotate=-90] (grayseal) at (26.0,0.0) {\Huge gray seal};
\node[rotate=-90] (orangutan) at (20.0,0.0) {\Huge orangutan};
\node[rotate=-90] (harborseal) at (27.0,0.0) {\Huge harbor seal};
\node[rotate=-90] (westerngreykangaroo) at (7.0,0.0) {\Huge western grey kangaroo};
\node[rotate=-90] (tiger) at (12.0,0.0) {\Huge tiger};
\node[rotate=-90] (lion) at (13.0,0.0) {\Huge lion};
\node[rotate=-90] (horse) at (11.0,0.0) {\Huge horse};
\node[rotate=-90] (easterngreykangaroo) at (8.0,0.0) {\Huge eastern grey kangaroo};
\node[rotate=-90] (longbeakedechidna) at (23.0,0.0) {\Huge long beaked echidna};
\node[rotate=-90] (gibbon) at (21.0,0.0) {\Huge gibbon};
\node[rotate=-90] (pygmychimpanzee) at (16.0,0.0) {\Huge pygmy chimpanzee};
\node[rotate=-90] (wildyak) at (2.0,0.0) {\Huge wild yak};
\node[rotate=-90] (humans) at (17.0,0.0) {\Huge humans};
\node[rotate=-90] (wildyakdomesticyak) at (2.5,1.0) {\Huge };
\draw  (wildyak) |- (wildyakdomesticyak.center);
\draw  (domesticyak) |- (wildyakdomesticyak.center);
\node[rotate=-90] (chimpanzeepygmychimpanzee) at (15.5,1.5474353693785516) {\Huge };
\draw  (chimpanzee) |- (chimpanzeepygmychimpanzee.center);
\draw  (pygmychimpanzee) |- (chimpanzeepygmychimpanzee.center);
\node[rotate=-90] (wildyakdomesticyakcow) at (3.25,1.7099542988763266) {\Huge };
\draw  (wildyakdomesticyak.center) |- (wildyakdomesticyakcow.center);
\draw  (cow) |- (wildyakdomesticyakcow.center);
\node[rotate=-90] (sumatranorangutanorangutan) at (19.5,1.7822749604559527) {\Huge };
\draw  (sumatranorangutan) |- (sumatranorangutanorangutan.center);
\draw  (orangutan) |- (sumatranorangutanorangutan.center);
\node[rotate=-90] (westerngreykangarooeasterngreykangaroo) at (7.5,1.8712224464724487) {\Huge };
\draw  (westerngreykangaroo) |- (westerngreykangarooeasterngreykangaroo.center);
\draw  (easterngreykangaroo) |- (westerngreykangarooeasterngreykangaroo.center);
\node[rotate=-90] (chimpanzeepygmychimpanzeehumans) at (16.25,1.9196748298583954) {\Huge };
\draw  (chimpanzeepygmychimpanzee.center) |- (chimpanzeepygmychimpanzeehumans.center);
\draw  (humans) |- (chimpanzeepygmychimpanzeehumans.center);
\node[rotate=-90] (chimpanzeepygmychimpanzeehumansgorilla) at (17.125,1.9463430769405567) {\Huge };
\draw  (chimpanzeepygmychimpanzeehumans.center) |- (chimpanzeepygmychimpanzeehumansgorilla.center);
\draw  (gorilla) |- (chimpanzeepygmychimpanzeehumansgorilla.center);
\node[rotate=-90] (bluewhalefinbackwhale) at (5.5,1.9594151585079098) {\Huge };
\draw  (bluewhale) |- (bluewhalefinbackwhale.center);
\draw  (finbackwhale) |- (bluewhalefinbackwhale.center);
\node[rotate=-90] (westerngreykangarooeasterngreykangaroowallaroo) at (8.25,1.9850575891212472) {\Huge };
\draw  (westerngreykangarooeasterngreykangaroo.center) |- (westerngreykangarooeasterngreykangaroowallaroo.center);
\draw  (wallaroo) |- (westerngreykangarooeasterngreykangaroowallaroo.center);
\node[rotate=-90] (lioncat) at (13.5,2.0523612710173538) {\Huge };
\draw  (lion) |- (lioncat.center);
\draw  (cat) |- (lioncat.center);
\node[rotate=-90] (zebrahorse) at (10.5,2.0641261125969406) {\Huge };
\draw  (zebra) |- (zebrahorse.center);
\draw  (horse) |- (zebrahorse.center);
\node[rotate=-90] (tigerlioncat) at (12.75,2.1042461063863658) {\Huge };
\draw  (tiger) |- (tigerlioncat.center);
\draw  (lioncat.center) |- (tigerlioncat.center);
\node[rotate=-90] (chimpanzeepygmychimpanzeehumansgorillasumatranorangutanorangutan) at (18.3125,2.164209571153923) {\Huge };
\draw  (chimpanzeepygmychimpanzeehumansgorilla.center) |- (chimpanzeepygmychimpanzeehumansgorillasumatranorangutanorangutan.center);
\draw  (sumatranorangutanorangutan.center) |- (chimpanzeepygmychimpanzeehumansgorillasumatranorangutanorangutan.center);
\node[rotate=-90] (chimpanzeepygmychimpanzeehumansgorillasumatranorangutanorangutangibbon) at (19.65625,2.1903207151576085) {\Huge };
\draw  (chimpanzeepygmychimpanzeehumansgorillasumatranorangutanorangutan.center) |- (chimpanzeepygmychimpanzeehumansgorillasumatranorangutanorangutangibbon.center);
\draw  (gibbon) |- (chimpanzeepygmychimpanzeehumansgorillasumatranorangutanorangutangibbon.center);
\node[rotate=-90] (bluewhalefinbackwhalewesterngreykangarooeasterngreykangaroowallaroo) at (6.875,2.2647968080590912) {\Huge };
\draw  (bluewhalefinbackwhale.center) |- (bluewhalefinbackwhalewesterngreykangarooeasterngreykangaroowallaroo.center);
\draw  (westerngreykangarooeasterngreykangaroowallaroo.center) |- (bluewhalefinbackwhalewesterngreykangarooeasterngreykangaroowallaroo.center);
\node[rotate=-90] (rathousemouse) at (0.5,2.269570086811749) {\Huge };
\draw  (rat) |- (rathousemouse.center);
\draw  (housemouse) |- (rathousemouse.center);
\node[rotate=-90] (rathousemousewildyakdomesticyakcow) at (1.875,2.320625257612033) {\Huge };
\draw  (rathousemouse.center) |- (rathousemousewildyakdomesticyakcow.center);
\draw  (wildyakdomesticyakcow.center) |- (rathousemousewildyakdomesticyakcow.center);
\node[rotate=-90] (bluewhalefinbackwhalewesterngreykangarooeasterngreykangaroowallaroozebrahorse) at (8.6875,2.3341086079493207) {\Huge };
\draw  (bluewhalefinbackwhalewesterngreykangarooeasterngreykangaroowallaroo.center) |- (bluewhalefinbackwhalewesterngreykangarooeasterngreykangaroowallaroozebrahorse.center);
\draw  (zebrahorse.center) |- (bluewhalefinbackwhalewesterngreykangarooeasterngreykangaroowallaroozebrahorse.center);
\node[rotate=-90] (bluewhalefinbackwhalewesterngreykangarooeasterngreykangaroowallaroozebrahorsetigerlioncat) at (10.71875,2.3341086079493207) {\Huge };
\draw  (bluewhalefinbackwhalewesterngreykangarooeasterngreykangaroowallaroozebrahorse.center) |- (bluewhalefinbackwhalewesterngreykangarooeasterngreykangaroowallaroozebrahorsetigerlioncat.center);
\draw  (tigerlioncat.center) |- (bluewhalefinbackwhalewesterngreykangarooeasterngreykangaroowallaroozebrahorsetigerlioncat.center);
\node[rotate=-90] (rathousemousewildyakdomesticyakcowbluewhalefinbackwhalewesterngreykangarooeasterngreykangaroowallaroozebrahorsetigerlioncat) at (6.296875,2.3385629582987004) {\Huge };
\draw  (rathousemousewildyakdomesticyakcow.center) |- (rathousemousewildyakdomesticyakcowbluewhalefinbackwhalewesterngreykangarooeasterngreykangaroowallaroozebrahorsetigerlioncat.center);
\draw  (bluewhalefinbackwhalewesterngreykangarooeasterngreykangaroowallaroozebrahorsetigerlioncat.center) |- (rathousemousewildyakdomesticyakcowbluewhalefinbackwhalewesterngreykangarooeasterngreykangaroowallaroozebrahorsetigerlioncat.center);
\node[rotate=-90] (rathousemousewildyakdomesticyakcowbluewhalefinbackwhalewesterngreykangarooeasterngreykangaroowallaroozebrahorsetigerlioncatchimpanzeepygmychimpanzeehumansgorillasumatranorangutanorangutangibbon) at (12.9765625,2.4031014794362715) {\Huge };
\draw  (rathousemousewildyakdomesticyakcowbluewhalefinbackwhalewesterngreykangarooeasterngreykangaroowallaroozebrahorsetigerlioncat.center) |- (rathousemousewildyakdomesticyakcowbluewhalefinbackwhalewesterngreykangarooeasterngreykangaroowallaroozebrahorsetigerlioncatchimpanzeepygmychimpanzeehumansgorillasumatranorangutanorangutangibbon.center);
\draw  (chimpanzeepygmychimpanzeehumansgorillasumatranorangutanorangutangibbon.center) |- (rathousemousewildyakdomesticyakcowbluewhalefinbackwhalewesterngreykangarooeasterngreykangaroowallaroozebrahorsetigerlioncatchimpanzeepygmychimpanzeehumansgorillasumatranorangutanorangutangibbon.center);
\node[rotate=-90] (rathousemousewildyakdomesticyakcowbluewhalefinbackwhalewesterngreykangarooeasterngreykangaroowallaroozebrahorsetigerlioncatchimpanzeepygmychimpanzeehumansgorillasumatranorangutanorangutangibbonwhiterhino) at (17.48828125,2.4196307813874824) {\Huge };
\draw  (rathousemousewildyakdomesticyakcowbluewhalefinbackwhalewesterngreykangarooeasterngreykangaroowallaroozebrahorsetigerlioncatchimpanzeepygmychimpanzeehumansgorillasumatranorangutanorangutangibbon.center) |- (rathousemousewildyakdomesticyakcowbluewhalefinbackwhalewesterngreykangarooeasterngreykangaroowallaroozebrahorsetigerlioncatchimpanzeepygmychimpanzeehumansgorillasumatranorangutanorangutangibbonwhiterhino.center);
\draw  (whiterhino) |- (rathousemousewildyakdomesticyakcowbluewhalefinbackwhalewesterngreykangarooeasterngreykangaroowallaroozebrahorsetigerlioncatchimpanzeepygmychimpanzeehumansgorillasumatranorangutanorangutangibbonwhiterhino.center);
\node[rotate=-90] (longbeakedechidnaplatypus) at (23.5,2.4358913022592628) {\Huge };
\draw  (longbeakedechidna) |- (longbeakedechidnaplatypus.center);
\draw  (platypus) |- (longbeakedechidnaplatypus.center);
\node[rotate=-90] (rathousemousewildyakdomesticyakcowbluewhalefinbackwhalewesterngreykangarooeasterngreykangaroowallaroozebrahorsetigerlioncatchimpanzeepygmychimpanzeehumansgorillasumatranorangutanorangutangibbonwhiterhinolongbeakedechidnaplatypus) at (20.494140625,2.4518916436057037) {\Huge };
\draw  (rathousemousewildyakdomesticyakcowbluewhalefinbackwhalewesterngreykangarooeasterngreykangaroowallaroozebrahorsetigerlioncatchimpanzeepygmychimpanzeehumansgorillasumatranorangutanorangutangibbonwhiterhino.center) |- (rathousemousewildyakdomesticyakcowbluewhalefinbackwhalewesterngreykangarooeasterngreykangaroowallaroozebrahorsetigerlioncatchimpanzeepygmychimpanzeehumansgorillasumatranorangutanorangutangibbonwhiterhinolongbeakedechidnaplatypus.center);
\draw  (longbeakedechidnaplatypus.center) |- (rathousemousewildyakdomesticyakcowbluewhalefinbackwhalewesterngreykangarooeasterngreykangaroowallaroozebrahorsetigerlioncatchimpanzeepygmychimpanzeehumansgorillasumatranorangutanorangutangibbonwhiterhinolongbeakedechidnaplatypus.center);
\node[rotate=-90] (rathousemousewildyakdomesticyakcowbluewhalefinbackwhalewesterngreykangarooeasterngreykangaroowallaroozebrahorsetigerlioncatchimpanzeepygmychimpanzeehumansgorillasumatranorangutanorangutangibbonwhiterhinolongbeakedechidnaplatypusopossum) at (22.7470703125,2.5888892442067126) {\Huge };
\draw  (rathousemousewildyakdomesticyakcowbluewhalefinbackwhalewesterngreykangarooeasterngreykangaroowallaroozebrahorsetigerlioncatchimpanzeepygmychimpanzeehumansgorillasumatranorangutanorangutangibbonwhiterhinolongbeakedechidnaplatypus.center) |- (rathousemousewildyakdomesticyakcowbluewhalefinbackwhalewesterngreykangarooeasterngreykangaroowallaroozebrahorsetigerlioncatchimpanzeepygmychimpanzeehumansgorillasumatranorangutanorangutangibbonwhiterhinolongbeakedechidnaplatypusopossum.center);
\draw  (opossum) |- (rathousemousewildyakdomesticyakcowbluewhalefinbackwhalewesterngreykangarooeasterngreykangaroowallaroozebrahorsetigerlioncatchimpanzeepygmychimpanzeehumansgorillasumatranorangutanorangutangibbonwhiterhinolongbeakedechidnaplatypusopossum.center);
\node[rotate=-90] (rathousemousewildyakdomesticyakcowbluewhalefinbackwhalewesterngreykangarooeasterngreykangaroowallaroozebrahorsetigerlioncatchimpanzeepygmychimpanzeehumansgorillasumatranorangutanorangutangibbonwhiterhinolongbeakedechidnaplatypusopossumgrayseal) at (24.37353515625,3.6943981119338827) {\Huge };
\draw  (rathousemousewildyakdomesticyakcowbluewhalefinbackwhalewesterngreykangarooeasterngreykangaroowallaroozebrahorsetigerlioncatchimpanzeepygmychimpanzeehumansgorillasumatranorangutanorangutangibbonwhiterhinolongbeakedechidnaplatypusopossum.center) |- (rathousemousewildyakdomesticyakcowbluewhalefinbackwhalewesterngreykangarooeasterngreykangaroowallaroozebrahorsetigerlioncatchimpanzeepygmychimpanzeehumansgorillasumatranorangutanorangutangibbonwhiterhinolongbeakedechidnaplatypusopossumgrayseal.center);
\draw  (grayseal) |- (rathousemousewildyakdomesticyakcowbluewhalefinbackwhalewesterngreykangarooeasterngreykangaroowallaroozebrahorsetigerlioncatchimpanzeepygmychimpanzeehumansgorillasumatranorangutanorangutangibbonwhiterhinolongbeakedechidnaplatypusopossumgrayseal.center);
\node[rotate=-90] (rathousemousewildyakdomesticyakcowbluewhalefinbackwhalewesterngreykangarooeasterngreykangaroowallaroozebrahorsetigerlioncatchimpanzeepygmychimpanzeehumansgorillasumatranorangutanorangutangibbonwhiterhinolongbeakedechidnaplatypusopossumgraysealharborseal) at (25.686767578125,3.900676298724979) {\Huge };
\draw  (rathousemousewildyakdomesticyakcowbluewhalefinbackwhalewesterngreykangarooeasterngreykangaroowallaroozebrahorsetigerlioncatchimpanzeepygmychimpanzeehumansgorillasumatranorangutanorangutangibbonwhiterhinolongbeakedechidnaplatypusopossumgrayseal.center) |- (rathousemousewildyakdomesticyakcowbluewhalefinbackwhalewesterngreykangarooeasterngreykangaroowallaroozebrahorsetigerlioncatchimpanzeepygmychimpanzeehumansgorillasumatranorangutanorangutangibbonwhiterhinolongbeakedechidnaplatypusopossumgraysealharborseal.center);
\draw  (harborseal) |- (rathousemousewildyakdomesticyakcowbluewhalefinbackwhalewesterngreykangarooeasterngreykangaroowallaroozebrahorsetigerlioncatchimpanzeepygmychimpanzeehumansgorillasumatranorangutanorangutangibbonwhiterhinolongbeakedechidnaplatypusopossumgraysealharborseal.center);
\end{tikzpicture}
    \end{turn}
}
\caption{EBWT} \label{fig:mtdna_ebwt}
\end{subfigure}
\caption{Single-Link Clustering result on mitochondrial mtDNA data. } 
\label{fig:mtdna_all}
\end{figure}

First, we will evaluate mtDNA data on 28 different species, and use Single-Link Clustering to produce a dendrogram of the species based on their mtDNA. The results are shown in \autoref{fig:mtdna_all}, with both EBWT and BWMD taking under 1 second to perform the clustering. Our goal is not to fully evaluate the quality of each dendrogram, but to show that both methods produce reasonable results in this case, which may be of interest to researchers in bioinformatics. Both EBWT and BWMD do reasonably well at this task, with differing mistakes, advantages, and disadvantages. 

EBWT gets most base level groups correct (e.g., lion, tiger, cat in one group, primates grouped together). There is a failure to properly group the harbor and gray seals as related to each other, and instead act as outliers which SLINK is forced into a cluster at the end at higher cost. EBWT also fails to group the white rhino with other members of the Ferungulates family (e.g., the horse and zebra would have been closest members) \cite{Cao1998}. BWMD was able to correctly pair the seals and placed white rhino with a larger family of Ferungulates (closest to horse, which is correct, and with the cows and yak which are members). But BWMD failed to place the mouse and rat together, and dispersed the zebra and lion from their more appropriate neighbors.

Results with both methods are reasonable. However, the mtDNA task is an easier task. All of the mtDNA sequences are of similar sizes, with the western grey kangaroo being shortest at 15 KB in length and the lion being longest at 17 KB. 
Our theoretical results in \autoref{sec:theoretical_results} would indicate that we may see more significant issues if we had sequences of varying length. We explore this with a dataset of chromosome genomic scaffold DNA for a subset of the species evaluated. We selected one whole scaffold DNA sequence from a random chromosome for the 11 species where this was available. We selected "unplaced genomic scaffold" sequences for 3 remaining species (the yaks and tiger), which is a much shorter and incomplete amount of data. This gives us a minimum sequence size of 22 KB and a maximum of 33 MB. 

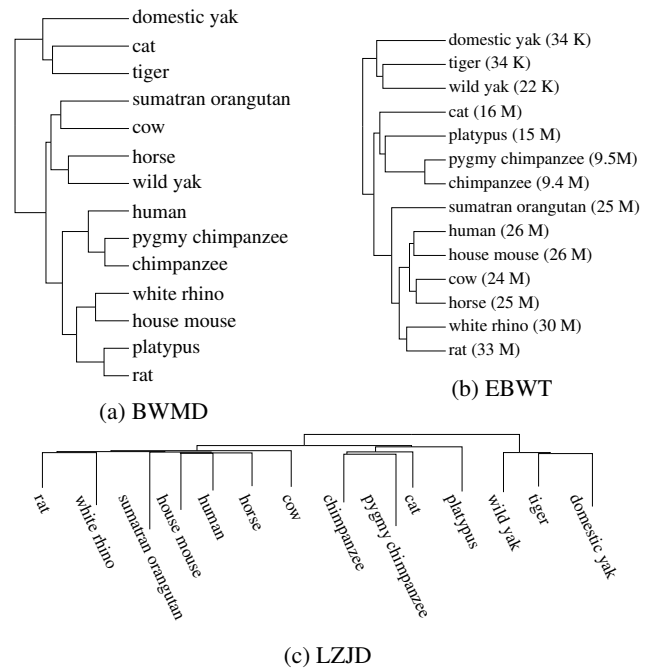
\begin{figure}[!htb]
\centering
\begin{subfigure}{0.45\columnwidth}
\centering
\adjustbox{max width=\columnwidth}{%
    \begin{turn}{90}
\begin{tikzpicture}[sloped]
\tikzset{anchor=west}
\node[rotate=-90] (platypus) at (1.0,0.0) {\Huge platypus};
\node[rotate=-90] (rat) at (0.0,0.0) {\Huge rat};
\node[rotate=-90] (cow) at (9.0,0.0) {\Huge cow};
\node[rotate=-90] (tiger) at (11.0,0.0) {\Huge tiger};
\node[rotate=-90] (horse) at (8.0,0.0) {\Huge horse};
\node[rotate=-90] (housemouse) at (2.0,0.0) {\Huge house mouse};
\node[rotate=-90] (whiterhino) at (3.0,0.0) {\Huge white rhino};
\node[rotate=-90] (cat) at (12.0,0.0) {\Huge cat};
\node[rotate=-90] (sumatranorangutan) at (10.0,0.0) {\Huge sumatran orangutan};
\node[rotate=-90] (domesticyak) at (13.0,0.0) {\Huge domestic yak};
\node[rotate=-90] (chimpanzee) at (4.0,0.0) {\Huge chimpanzee};
\node[rotate=-90] (pygmychimpanzee) at (5.0,0.0) {\Huge pygmy chimpanzee};
\node[rotate=-90] (human) at (6.0,0.0) {\Huge human};
\node[rotate=-90] (wildyak) at (7.0,0.0) {\Huge wild yak};
\node[rotate=-90] (chimpanzeepygmychimpanzee) at (4.5,1.0) {};
\draw  (chimpanzee) |- (chimpanzeepygmychimpanzee.center);
\draw  (pygmychimpanzee) |- (chimpanzeepygmychimpanzee.center);
\node[rotate=-90] (ratplatypus) at (0.5,1.0303651742646316) {};
\draw  (rat) |- (ratplatypus.center);
\draw  (platypus) |- (ratplatypus.center);
\node[rotate=-90] (housemousewhiterhino) at (2.5,1.3393655161434752) {};
\draw  (housemouse) |- (housemousewhiterhino.center);
\draw  (whiterhino) |- (housemousewhiterhino.center);
\node[rotate=-90] (chimpanzeepygmychimpanzeehuman) at (5.25,1.6018285568489243) {};
\draw  (chimpanzeepygmychimpanzee.center) |- (chimpanzeepygmychimpanzeehuman.center);
\draw  (human) |- (chimpanzeepygmychimpanzeehuman.center);
\node[rotate=-90] (ratplatypushousemousewhiterhino) at (1.5,2.02983137560803) {};
\draw  (ratplatypus.center) |- (ratplatypushousemousewhiterhino.center);
\draw  (housemousewhiterhino.center) |- (ratplatypushousemousewhiterhino.center);
\node[rotate=-90] (wildyakhorse) at (7.5,2.326450474402672) {};
\draw  (wildyak) |- (wildyakhorse.center);
\draw  (horse) |- (wildyakhorse.center);
\node[rotate=-90] (ratplatypushousemousewhiterhinochimpanzeepygmychimpanzeehuman) at (3.375,2.5454123990359556) {};
\draw  (ratplatypushousemousewhiterhino.center) |- (ratplatypushousemousewhiterhinochimpanzeepygmychimpanzeehuman.center);
\draw  (chimpanzeepygmychimpanzeehuman.center) |- (ratplatypushousemousewhiterhinochimpanzeepygmychimpanzeehuman.center);
\node[rotate=-90] (cowsumatranorangutan) at (9.5,2.6040188419507335) {};
\draw  (cow) |- (cowsumatranorangutan.center);
\draw  (sumatranorangutan) |- (cowsumatranorangutan.center);
\node[rotate=-90] (tigercat) at (11.5,2.904471198374458) {};
\draw  (tiger) |- (tigercat.center);
\draw  (cat) |- (tigercat.center);
\node[rotate=-90] (wildyakhorsecowsumatranorangutan) at (8.5,2.9611110647077195) {};
\draw  (wildyakhorse.center) |- (wildyakhorsecowsumatranorangutan.center);
\draw  (cowsumatranorangutan.center) |- (wildyakhorsecowsumatranorangutan.center);
\node[rotate=-90] (ratplatypushousemousewhiterhinochimpanzeepygmychimpanzeehumanwildyakhorsecowsumatranorangutan) at (5.9375,3.1433386987434533) {};
\draw  (ratplatypushousemousewhiterhinochimpanzeepygmychimpanzeehuman.center) |- (ratplatypushousemousewhiterhinochimpanzeepygmychimpanzeehumanwildyakhorsecowsumatranorangutan.center);
\draw  (wildyakhorsecowsumatranorangutan.center) |- (ratplatypushousemousewhiterhinochimpanzeepygmychimpanzeehumanwildyakhorsecowsumatranorangutan.center);
\node[rotate=-90] (tigercatdomesticyak) at (12.25,3.258565274112426) {};
\draw  (tigercat.center) |- (tigercatdomesticyak.center);
\draw  (domesticyak) |- (tigercatdomesticyak.center);
\node[rotate=-90] (ratplatypushousemousewhiterhinochimpanzeepygmychimpanzeehumanwildyakhorsecowsumatranorangutantigercatdomesticyak) at (9.09375,4.268340115108067) {};
\draw  (ratplatypushousemousewhiterhinochimpanzeepygmychimpanzeehumanwildyakhorsecowsumatranorangutan.center) |- (ratplatypushousemousewhiterhinochimpanzeepygmychimpanzeehumanwildyakhorsecowsumatranorangutantigercatdomesticyak.center);
\draw  (tigercatdomesticyak.center) |- (ratplatypushousemousewhiterhinochimpanzeepygmychimpanzeehumanwildyakhorsecowsumatranorangutantigercatdomesticyak.center);
\end{tikzpicture}
    \end{turn}
}
\caption{BWMD} \label{fig:scaffold_bwmd}
\end{subfigure}
\hspace*{\fill} %
\begin{subfigure}{0.45\columnwidth}
\centering
\adjustbox{max width=\columnwidth}{%
    \begin{turn}{90}
\begin{tikzpicture}[sloped]
\tikzset{anchor=west}
\node[rotate=-90] (platypus) at (9.0,0.0) {\Huge platypus (15 M) };
\node[rotate=-90] (rat) at (0.0,0.0) {\Huge rat (33 M) };
\node[rotate=-90] (cow) at (3.0,0.0) {\Huge cow (24 M)};
\node[rotate=-90] (horse) at (2.0,0.0) {\Huge horse (25 M)};
\node[rotate=-90] (tiger) at (12.0,0.0) {\Huge tiger (34 K)};
\node[rotate=-90] (housemouse) at (4.0,0.0) {\Huge house mouse (26 M)};
\node[rotate=-90] (whiterhino) at (1.0,0.0) {\Huge white rhino (30 M) };
\node[rotate=-90] (cat) at (10.0,0.0) {\Huge cat (16 M)};
\node[rotate=-90] (sumatranorangutan) at (6.0,0.0) {\Huge sumatran orangutan (25 M)};
\node[rotate=-90] (domesticyak) at (13.0,0.0) {\Huge domestic yak (34 K)};
\node[rotate=-90] (chimpanzee) at (7.0,0.0) {\Huge chimpanzee (9.4 M)};
\node[rotate=-90] (pygmychimpanzee) at (8.0,0.0) {\Huge pygmy chimpanzee (9.5M) };
\node[rotate=-90] (human) at (5.0,0.0) {\Huge human (26 M)};
\node[rotate=-90] (wildyak) at (11.0,0.0) {\Huge wild yak (22 K)};
\node[rotate=-90] (chimpanzeepygmychimpanzee) at (7.5,1.0) {};
\draw  (chimpanzee) |- (chimpanzeepygmychimpanzee.center);
\draw  (pygmychimpanzee) |- (chimpanzeepygmychimpanzee.center);
\node[rotate=-90] (cowhorse) at (2.5,1.3593740009078974) {};
\draw  (cow) |- (cowhorse.center);
\draw  (horse) |- (cowhorse.center);
\node[rotate=-90] (housemousehuman) at (4.5,1.4581088415935863) {};
\draw  (housemouse) |- (housemousehuman.center);
\draw  (human) |- (housemousehuman.center);
\node[rotate=-90] (cowhorsehousemousehuman) at (3.5,1.6411423955696307) {};
\draw  (cowhorse.center) |- (cowhorsehousemousehuman.center);
\draw  (housemousehuman.center) |- (cowhorsehousemousehuman.center);
\node[rotate=-90] (ratwhiterhino) at (0.5,1.7585299398227972) {};
\draw  (rat) |- (ratwhiterhino.center);
\draw  (whiterhino) |- (ratwhiterhino.center);
\node[rotate=-90] (ratwhiterhinocowhorsehousemousehuman) at (2.0,2.070634443262673) {};
\draw  (ratwhiterhino.center) |- (ratwhiterhinocowhorsehousemousehuman.center);
\draw  (cowhorsehousemousehuman.center) |- (ratwhiterhinocowhorsehousemousehuman.center);
\node[rotate=-90] (ratwhiterhinocowhorsehousemousehumansumatranorangutan) at (4.0,2.386294361119891) {};
\draw  (ratwhiterhinocowhorsehousemousehuman.center) |- (ratwhiterhinocowhorsehousemousehumansumatranorangutan.center);
\draw  (sumatranorangutan) |- (ratwhiterhinocowhorsehousemousehumansumatranorangutan.center);
\node[rotate=-90] (chimpanzeepygmychimpanzeeplatypus) at (8.25,2.656940246419104) {};
\draw  (chimpanzeepygmychimpanzee.center) |- (chimpanzeepygmychimpanzeeplatypus.center);
\draw  (platypus) |- (chimpanzeepygmychimpanzeeplatypus.center);
\node[rotate=-90] (wildyaktiger) at (11.5,2.7550581023776273) {};
\draw  (wildyak) |- (wildyaktiger.center);
\draw  (tiger) |- (wildyaktiger.center);
\node[rotate=-90] (chimpanzeepygmychimpanzeeplatypuscat) at (9.125,2.882829166650506) {};
\draw  (chimpanzeepygmychimpanzeeplatypus.center) |- (chimpanzeepygmychimpanzeeplatypuscat.center);
\draw  (cat) |- (chimpanzeepygmychimpanzeeplatypuscat.center);
\node[rotate=-90] (wildyaktigerdomesticyak) at (12.25,3.013099593543114) {};
\draw  (wildyaktiger.center) |- (wildyaktigerdomesticyak.center);
\draw  (domesticyak) |- (wildyaktigerdomesticyak.center);
\node[rotate=-90] (ratwhiterhinocowhorsehousemousehumansumatranorangutanchimpanzeepygmychimpanzeeplatypuscat) at (6.5625,3.1319707997487285) {};
\draw  (ratwhiterhinocowhorsehousemousehumansumatranorangutan.center) |- (ratwhiterhinocowhorsehousemousehumansumatranorangutanchimpanzeepygmychimpanzeeplatypuscat.center);
\draw  (chimpanzeepygmychimpanzeeplatypuscat.center) |- (ratwhiterhinocowhorsehousemousehumansumatranorangutanchimpanzeepygmychimpanzeeplatypuscat.center);
\node[rotate=-90] (ratwhiterhinocowhorsehousemousehumansumatranorangutanchimpanzeepygmychimpanzeeplatypuscatwildyaktigerdomesticyak) at (9.40625,3.599469168593897) {};
\draw  (ratwhiterhinocowhorsehousemousehumansumatranorangutanchimpanzeepygmychimpanzeeplatypuscat.center) |- (ratwhiterhinocowhorsehousemousehumansumatranorangutanchimpanzeepygmychimpanzeeplatypuscatwildyaktigerdomesticyak.center);
\draw  (wildyaktigerdomesticyak.center) |- (ratwhiterhinocowhorsehousemousehumansumatranorangutanchimpanzeepygmychimpanzeeplatypuscatwildyaktigerdomesticyak.center);
\end{tikzpicture}
    \end{turn}
}
\caption{EBWT} \label{fig:scaffold_ebwt}
\end{subfigure}

\begin{subfigure}{0.95\columnwidth}
\centering
\adjustbox{max width=\columnwidth}{%
\begin{tikzpicture}[sloped]
\tikzset{anchor=west}
\node[rotate=-65] (platypus) at (10.0,0.0) {\Large platypus};
\node[rotate=-65] (rat) at (0.0,0.0) {\Large rat};
\node[rotate=-65] (cow) at (6.0,0.0) {\Large cow};
\node[rotate=-65] (tiger) at (12.0,0.0) {\Large tiger};
\node[rotate=-65] (horse) at (5.0,0.0) {\Large horse};
\node[rotate=-65] (housemouse) at (3.0,0.0) {\Large house mouse};
\node[rotate=-65] (whiterhino) at (1.0,0.0) {\Large white rhino};
\node[rotate=-65] (cat) at (9.0,0.0) {\Large cat};
\node[rotate=-65] (sumatranorangutan) at (2.0,0.0) {\Large sumatran orangutan};
\node[rotate=-65] (domesticyak) at (13.0,0.0) {\Large domestic yak};
\node[rotate=-65] (chimpanzee) at (7.0,0.0) {\Large chimpanzee};
\node[rotate=-65] (pygmychimpanzee) at (8.0,0.0) {\Large pygmy chimpanzee};
\node[rotate=-65] (human) at (4.0,0.0) {\Large human};
\node[rotate=-65] (wildyak) at (11.0,0.0) {\Large wild yak};
\node[rotate=-65] (chimpanzeepygmychimpanzee) at (7.5,1.0) {\Large };
\draw  (chimpanzee) |- (chimpanzeepygmychimpanzee.center);
\draw  (pygmychimpanzee) |- (chimpanzeepygmychimpanzee.center);
\node[rotate=-90] (tigerdomesticyak) at (12.5,1.027160471624535) {\Large };
\draw  (tiger) |- (tigerdomesticyak.center);
\draw  (domesticyak) |- (tigerdomesticyak.center);
\node[rotate=-90] (housemousehuman) at (3.5,1.0359105327926805) {\Large };
\draw  (housemouse) |- (housemousehuman.center);
\draw  (human) |- (housemousehuman.center);
\node[rotate=-90] (housemousehumanhorse) at (4.25,1.0430914288877517) {\Large };
\draw  (housemousehuman.center) |- (housemousehumanhorse.center);
\draw  (horse) |- (housemousehumanhorse.center);
\node[rotate=-90] (ratwhiterhino) at (0.5,1.0501739837464532) {\Large };
\draw  (rat) |- (ratwhiterhino.center);
\draw  (whiterhino) |- (ratwhiterhino.center);
\node[rotate=-90] (wildyaktigerdomesticyak) at (11.75,1.0515788868691065) {\Large };
\draw  (wildyak) |- (wildyaktigerdomesticyak.center);
\draw  (tigerdomesticyak.center) |- (wildyaktigerdomesticyak.center);
\node[rotate=-90] (sumatranorangutanhousemousehumanhorse) at (3.125,1.0640527234145338) {\Large };
\draw  (sumatranorangutan) |- (sumatranorangutanhousemousehumanhorse.center);
\draw  (housemousehumanhorse.center) |- (sumatranorangutanhousemousehumanhorse.center);
\node[rotate=-90] (chimpanzeepygmychimpanzeecat) at (8.25,1.0681438416049813) {\Large };
\draw  (chimpanzeepygmychimpanzee.center) |- (chimpanzeepygmychimpanzeecat.center);
\draw  (cat) |- (chimpanzeepygmychimpanzeecat.center);
\node[rotate=-90] (ratwhiterhinosumatranorangutanhousemousehumanhorse) at (1.8125,1.074889940763729) {\Large };
\draw  (ratwhiterhino.center) |- (ratwhiterhinosumatranorangutanhousemousehumanhorse.center);
\draw  (sumatranorangutanhousemousehumanhorse.center) |- (ratwhiterhinosumatranorangutanhousemousehumanhorse.center);
\node[rotate=-90] (ratwhiterhinosumatranorangutanhousemousehumanhorsecow) at (3.90625,1.0971732168565844) {\Large };
\draw  (ratwhiterhinosumatranorangutanhousemousehumanhorse.center) |- (ratwhiterhinosumatranorangutanhousemousehumanhorsecow.center);
\draw  (cow) |- (ratwhiterhinosumatranorangutanhousemousehumanhorsecow.center);
\node[rotate=-90] (chimpanzeepygmychimpanzeecatplatypus) at (9.125,1.191263559321234) {\Large };
\draw  (chimpanzeepygmychimpanzeecat.center) |- (chimpanzeepygmychimpanzeecatplatypus.center);
\draw  (platypus) |- (chimpanzeepygmychimpanzeecatplatypus.center);
\node[rotate=-90] (ratwhiterhinosumatranorangutanhousemousehumanhorsecowchimpanzeepygmychimpanzeecatplatypus) at (6.515625,1.2218176875997924) {\Large };
\draw  (ratwhiterhinosumatranorangutanhousemousehumanhorsecow.center) |- (ratwhiterhinosumatranorangutanhousemousehumanhorsecowchimpanzeepygmychimpanzeecatplatypus.center);
\draw  (chimpanzeepygmychimpanzeecatplatypus.center) |- (ratwhiterhinosumatranorangutanhousemousehumanhorsecowchimpanzeepygmychimpanzeecatplatypus.center);
\node[rotate=-90] (ratwhiterhinosumatranorangutanhousemousehumanhorsecowchimpanzeepygmychimpanzeecatplatypuswildyaktigerdomesticyak) at (9.1328125,1.4929274836750976) {\Large };
\draw  (ratwhiterhinosumatranorangutanhousemousehumanhorsecowchimpanzeepygmychimpanzeecatplatypus.center) |- (ratwhiterhinosumatranorangutanhousemousehumanhorsecowchimpanzeepygmychimpanzeecatplatypuswildyaktigerdomesticyak.center);
\draw  (wildyaktigerdomesticyak.center) |- (ratwhiterhinosumatranorangutanhousemousehumanhorsecowchimpanzeepygmychimpanzeecatplatypuswildyaktigerdomesticyak.center);
\end{tikzpicture}
}
\caption{LZJD} \label{fig:scaffold_lzjd}
\end{subfigure}

\caption{Single-Link Clustering on Genomic Scaffolding. } 
\label{fig:scaffold_all}
\end{figure}

The SLINK results are shown in \autoref{fig:scaffold_all}. At a base level, the cost to perform EBWT and its scalability issues are more pronounced. BWMD takes only 47 seconds to perform SLINK clustering. EBWT took 28 minutes, making it over 35$\times$ slower. When plotting the EBWT dendrogram in \autoref{fig:scaffold_ebwt}, we include the size of the DNA sequence in parentheses. When organized in this way, it becomes clear that the EBWT clustering is degenerate, and corresponds exactly to file-size, rather than content. 

In contrast, the BWMD in \autoref{fig:scaffold_bwmd} produces reasonable groupings despite disparate sequence lengths. For example, (full scaffold) cat and (unplaced and incomplete) tiger are correctly grouped despite the cat sequence being 450$\times$ longer. The BWMD results are not perfect: the orangutan and domestic yak were placed farther from the other (well-grouped) primates and ferungulates respectively. Overall we can see that groups which should be placed near each other are, 
without degrading to sequence length information. 

We include LZJD ( \autoref{fig:scaffold_lzjd} )in this scenario, to demonstrate that BWMD has increased value over other alternatives. Here we can see that LZJD suffers the same failure as EBWT in placing the three smallest partial segments into a single cluster. With the exception of correctly grouping the two chimpanzee species, the LZJD dendrogram is degenerate. 

These results are in line with our 
theory 
derived in \autoref{sec:theoretical_results}. 
When working with sequences of homogeneous length, EBWT performed well. 
But BWMD is able to handle disparate sequence lengths reasonably well, where EBWT degrades to grouping by sequence length rather than content. 

BWMD's ability to handle the original mtDNA data, as well as substantially better results with the irregular-sized scaffold DNA, is made more impressive by the fact that BWMD is encoding everything into $\mathbb{R}^{16}$ due to the small alphabet $|\Sigma|=4$. This is a reduction in storage cost by a factor of up to 515.6$\times$, and allows for more flexibility in creating a larger and searchable index using BWMD.

\subsection{A note on BMWD's Disadvantage}

It is also worth noting that, from an information encoding perspective, BWMD is at a disadvantage in this testing over DNA data. EBWT is dimensionless, and has the representation capacity of the merged sorting of two different strings, meaning its representational capacity is a function of the sequence length under consideration. BWMD encodes each sequence into a fixed-length feature vector of size $|\Sigma|^2$. Since we are working with DNA data, the alphabet $\Sigma=\{A, T, C, G\}$ is quite small. As such all DNA sequences in these experiments, up to 33 MB in size, are being embedded into a 16-dimensional space. 
BWMD's ability to match or outperform EBWT
means we must be doing a significantly better job at leveraging the first-order information expressed by the Burrows Wheeler Transform.

\section{Malware Results} \label{sec:malware}

It can be difficult to reliably parse malicious software as malware authors may intentionally violate rules and format standards. Compression based similarity measures are useful in this case as they allow us to avoid these complex parsing issues. In this section we will look at the classification accuracy of BWMD and LZJD for several datasets. Our expectation is that LZJD will have better classification accuracy, due to Lempel-Ziv compressors being more effective than those based on Burrows Wheeler. However, we find that BWMD has significant advantages when clustering is the goal accuracy by up to 65.6$\times$, and is 24$\times$ faster to search large corpora by obtaining sub-linear scaling  with no loss in accuracy. Since LZJD cannot achieve this, this advantage will only increase with corpus size. 

Many prior works have looked at malware classification and clustering by processing raw bytes, due to the difficulty of parsing malware. 
We will compare with the seminal BitShred algorithm \cite{Jang2011} for clustering. Other similar byte based distance functions such as Ssdeep and Sdhash where evaluated, but founds to have degenerate performance on our smallest and easiest corpora, which can be found in appendix \autoref{sec:malware_app}. Compression distances such as EBWT and NCD are not evaluated in this section, because it would require multiple compute months for our smallest dataset, and simply cannot scale to the size of our malware corpora. 

\begin{table}[!ht]
\caption{Malware datasets used in experiments.  }
\label{tbl:dataset_summary}
\begin{adjustbox}{max size={1.00\columnwidth}{0.85\textheight}}
\begin{tabular}{@{}lrcrrr@{}}
\toprule
Dataset      & \multicolumn{1}{c}{Avg Size} & \multicolumn{1}{c}{\# Classes} & \multicolumn{1}{c}{Train} & \multicolumn{1}{c}{Test} & Storage Size\\ \midrule
EMBER        & 1.17 MB & 2 & 600,000                   & 200,000         & 936 GB         \\
VS20F        & 705 KB & 20 & 160,000                   & 40,000       & 141 GB              \\
Kaggle Bytes & 4.67 MB & 9  & 10,868                   & ---       & 50.8 GB              \\ 
Kaggle ASM   & 13.5 MB & 9  & 10,868                   & ---       & 147 GB              \\ 
Drebin APK   & 1.37 MB & 20 & 4,664                   & ---       & 6.4 GB              \\ 
Drebin TAR   & 1.84 MB & 20 & 4,664                   & ---       & 8.6 GB              \\ 
\bottomrule
\end{tabular}
\end{adjustbox}
\end{table}

For our evaluation, we will use several datasets summarized in \autoref{tbl:dataset_summary}. The EMBER dataset \cite{Anderson2018} pertains to a binary classification problem of "benign vs malicious" for Windows executables. Because there are only two classes, clustering results will not be evaluated on this corpus. However it is by far the largest corpus, allowing us to explore the scalability of our algorithms.  The raw files can be obtained from VirusTotal (\url{www.virustotal.com})
and are nearly 1TB total. Our remaining datasets will be multi-class problems where each sample is a member of a specific malware family. We will use these to evaluate both classification accuracy, as well as accuracy in clustering with respect to the class labels. 
Using VirusShare \cite{VirusShare} we create another Windows based dataset with 20 malware families. The families were determined using VirusTotal and the AVClass tool which determines a single canonical malware family label based on multiple Anti-Virus outputs \cite{Sebastian2016}. We select the 20 most populous families, and use 7,000 examples for training and 3,000 for testing. The last four datasets we use are evaluated in "two forms", following prior work \cite{raff_lzjd_2017}. The Kaggle datasets are from a 2015 Kaggle competition sponsored by Microsoft \cite{microsoft_kaggle_2015}. In the "Bytes" version our algorithms are run on the raw malware binary, and in the "ASM" version the output of IDA-Pro's disassembler is used instead. From the Drebin corpus \cite{Arp2014} we use the 20 most populous families,  where the "APK" version is the raw bytes of the Android APK (essentially a Zip file with light compression), and the "TAR" version which unpacks the APK and recombines all content into a single tar file.

\subsection{Malware Classification}

We begin our analysis by looking at nearest neighbor classification performance of various methods. The performance of each algorithm under this scenario gives us insight not only to its utility, but how effective it would be for analysts in finding similar malware. Utility in this scenario requires both high accuracy \textit{and} computational efficiency, as malware corpora are often measured in the terabyte to petabyte range. 

\subsubsection{Small Scale Malware Classification}
\label{sec:malware_app}

The Kaggle and Drebin corpora are considerably smaller in size, allowing us to test a wider selection of methods against them. In the below table we use balanced accuracy, where the weights of each file are adjusted so that the total weight of each class is equal, because malware families are not evenly distributed in each corpus. 

\begin{table}[!h]
\caption{Balanced Accuracy results for 1-NN classification on each dataset. Results show mean 10-Fold Cross Validation accuracy (standard deviation in parentheses). Best results in \textbf{bold}, second best in \textit{italics}.  }
\label{tbl:family_classification}
\adjustbox{max width=\columnwidth}{%
\begin{tabular}{@{}lccccc@{}}
\toprule
Dataset      & Ssdeep     & Sdhash     & BitShred & BWMD                & LZJD                \\ \midrule
Kaggle Bytes & 38.4 (1.4) & 60.2 (2.3) & 43.7 (1.9) & \textit{96.4 (2.2)} & \textbf{97.9 (1.4)} \\
Kaggle ASM   & 26.6 (2.2) & 28.8 (1.3) & 36.9 (1.6) & \textit{97.0 (1.8)} & \textbf{97.1 (1.7)} \\
Drebin APK   & 13.6 (1.6) & 5.8 (0.5)  & \textit{58.3 (3.9)} & 55.3 (4.4) & \textbf{81.0 (4.0)} \\
Drebin TAR   & 24.2 (2.9) & 8.3 (1.2)  & 65.1 (3.7)  & \textit{76.3 (3.6)} & \textbf{87.9 (2.1)} \\
\bottomrule
\end{tabular}
}
\end{table}

We can see from these results that the compression-based approaches, LZJD and BWMD, generally outperform other alternatives by a wide margin. As was expected, LZJD has higher accuracy that BWMD, since LZJD is based on a more effective compression algorithm. 
While this is a slight weakness of BWMD, its advantage comes in being orders of magnitude faster,
as we will show in the large scale testing in the next section. 
This makes it the only method usable for larger-scale corpora. This also shows that Ssdeep and Sdhash are simply not accurate enough to be considered for use, without regard to computational constraints. 

BWMD performed second best on every dataset, the only exception being a 3 point difference to the BitShred algorithm. However, BWMD outperformed BitShred by at least 11 points on all other datasets. Supporting our theoretical analysis in \autoref{sec:theoretical_results}, we also see hints that BWMD is better equipped to work with extremely long sequences.  Most notably, BWMD is the only method which had improved accuracy and reduced variance when moving from Kaggle Bytes (4.67 MB) to Kaggle ASM (13.5 MB). This suggests that the disassembly may be in a form that allows the BWT to better capture first-order dependencies for compression. 

The fact that BWMD has non-trivial accuracy on Drebin APK (random guessing is 5\%) is particularly impressive and worth noting. This is because the APK files are essentially Zip files with a standard structure, and the Zip compression format is a more effective one than most BWT based methods such as bzip. As such, that there is any first-order information exploitable for effective similarity search is impressive, and indicates the utility of BWMD in wider applications. 

\subsubsection{Large Scale Malware Classification}
\label{sec:larg_scale_mal}

On EMBER we use 9-Nearest Neighbors as our classifier so that we can compute meaningful values for the Area Under the ROC Curve (AUC) \cite{Bradley1997}.
In this malware detection context, AUC can be interpreted as an algorithm's ability to properly prioritize all malicious software above all benign software. This metric is useful for prioritizing work queues, and is therefore particularly pertinent. We evaluate only BWMD and LZJD due to computational constraints on this larger dataset. 

BWMD obtains an AUC of 98.3\%, where LZJD acheives a slightly better 99.7\% AUC. As expected, LZJD obtains a higher accuracy than BWMD, because LZJD is built upon a more effective compression algorithm. However, accuracy in isolation does not determine what method is best to use. Due to the large size of malware corpora, sub-linear scaling is needed to be useful for realistic sized datasets. 

BWMD does have an advantage over LZJD in its ability to scale to large corpora in an efficient manner. This is critical, since industry datasets routinely require comparisons to terabytes of files or more
\cite{Roussev2010}.
Prior work has tried, with limited success, to scale LZJD to larger corpora. Using an extension of the Vantage Point (VP) tree \cite{Yianilos1993}, only a 2.5x speedup over brute force search was achieved\cite{Raff2018_metric_index}. Because BWMD operates by embedding files into a Euclidean space, we can leverage specialized algorithms like the Dynamic Continuous Index (DCI) algorithm \cite{Li2017b} that only work for the Euclidean distance. DCI works by projecting the whole dataset down to different, random, embeddings, and allows obtaining the true nearest neighbors in a fast and efficient manner. 
LZJD is not compatible with such algorithms, resulting in BWMD being better equipped for this task.

\begin{figure}[!htb]
\begin{center}
\centering
\begin{tikzpicture}[]
\begin{axis}[
    xlabel=Training Set Size,
    ylabel=Time (ms),
    legend style={font=\small},
    legend pos=north west,
    legend style={at={(0.00,1.0)},anchor=north west, draw=none,fill=none},
    legend columns=2,
    xmode=log,
    ymode=log,
    height=5cm,
    width=\columnwidth,
    ]
    
    \addplot+[mark=none,line width=1.0pt] table [x=n, y=Brute, ] {ember_scaling.dat};
    
    \addplot+[mark=none,line width=1.0pt] table [x=n, y=VPMV, ] {ember_scaling.dat};
    
    \addplot+[mark=none,line width=1.0pt] table [x=n, y=LZJD_Brute, ] {ember_scaling.dat};
    
    \addplot+[mark=none,line width=1.0pt] table [x=n, y=LZJD_VPMV, ] {ember_scaling.dat};
    
    \addplot+[mark=none,dashed,line width=1.0pt] table [x=n, y=DCI, ] {ember_scaling.dat};

    \legend{BWMD Brute, BWMD VP, LZJD Brute, LZJD VP, BWMD DCI} %

\end{axis}

\end{tikzpicture}
\end{center}
\caption{Table shows 9-NN search retrieval speed on the Ember test set (in milliseconds, y-axis, log-scale) as the number of training points (x-axis, log-scale) increases. 
}
\label{fig:ember_search_scaling}
\end{figure}
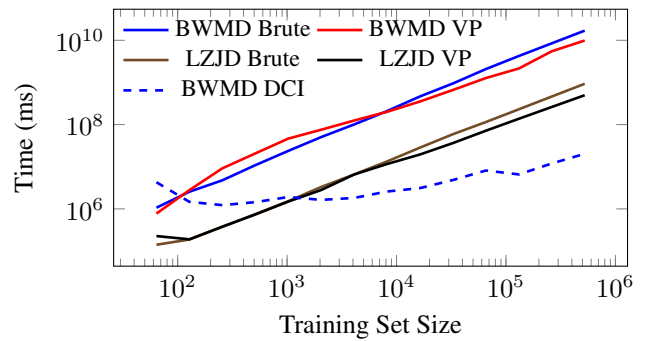

In \autoref{fig:ember_search_scaling} we compare the total query time of BWMD and LZJD under different indices, as the training set size increases from 64 files up to the full 600,000. We found the VP tree of minimum variance \cite{Raff2018_metric_index} performed best compared to other algorithms like KD and Cover-trees, and so only its results are included. In the dashed line we show BWMD accelerated with the Dynamic Continuous Index (DCI) algorithm \cite{Li2017b}.  

We can see that as the training corpus becomes larger, the VP trees are able to get small constant factor speedups, but are not able to reliably prune large portions of the search space. Because BWMD is in Euclidean space, it is the only method able to leverage the DCI algorithm and thus  able to get significant order-of-magnitude search speedups. This combination makes BWMD 24$\times$ faster than LZJD (5.6 CPU hours compared to \textit{5.6 days}), and 834$\times$ faster than BWMD with a brute force search. One can clearly see that DCI's scaling is sub-linear, and its advantage grows with the corpus size. This is obtained with no loss in accuracy on the Ember corpus, making BWMD the only effective approach for scaling to even larger corpora. 

\subsection{Malware Clustering}

In this section we will show that BWMD has significant advantages in terms of clustering malware into families. This benefit comes largely from BWMD mapping sequences into a Euclidean feature space, where we can leverage tried-and-true algorithms like k-means to perform fast and useful 
clustering.
LZJD is incompatible with k-means, 
and similar methods that require an explicit euclidean feature vector. As such LZJD, like BitShred, is constrained to distance based clustering methods like agglomerative clustering. This puts them at a significant disadvantage compared to BWMD. 

Evaluating the quality of our clustering results, we will consider three measures: Homogeneity, Completeness, and V-Measure, as introduced by \textcite{vMeasure} and using the class labels as ground truth cluster assignments. Homogeneity measures how well an algorithm does at making each found cluster as "pure" as possible (i.e., only one class in each cluster). Completeness measures how well an algorithm groups all examples of a class into as few clusters as possible (i.e., all examples of one class in only one cluster). V-Measure is the harmonic average of Homogeneity and Completeness. All three metrics are measured on the scale $[0,1]$, with 0 being worst, and 1 being the maximum score. 

In performing the clustering, we will test using $k=$ the true number of classes
and $k=10\times$ the true number of classes.
The former ($k=C$) is done to judge how well the clustering algorithms are able to recover the underlying ground truth. The latter 
($k=10\cdot C$)
is done as 
it corresponds best to how a malware analyst would desire to use these tools.%
It is easier to over-estimate the number of clusters than to predict the exact value of $k$, and by clustering an analyst would hope to reduce their workload by quickly checking that files in the same cluster are related, and then performing an in-depth analysis on only a few representatives from each cluster \cite{VanHoudnos2017}. For this reason, we consider Homogeneity the most important of the three measures, as it corresponds with how an analyst would use clustering, followed by V-Measure, and then Completeness.

\begin{table}[!h]
\caption{
Clustering performance of BWMD, LZJD, and BitShred. Best results shown in \textbf{bold}. 
}
\label{tbl:cluster_results}
\adjustbox{max width=\columnwidth}{%
\begin{tabular}{lccccccc}
\hline
                              &        & \multicolumn{3}{c}{$k=C$}            & \multicolumn{3}{c}{$k=10\cdot C$}              \\ \cline{3-8} 
Dataset                       & Metric & BWMD           & LZJD           & BitShred       & BWMD           & LZJD           & BitShred       \\ \hline
\multirow{3}{*}{Kaggle Bytes} & V-M    & \textbf{0.581} & 0.352          & 0.007          & \textbf{0.546} & 0.414          & 0.028          \\
                              & Homog  & \textbf{0.597} & 0.254          & 0.003          & \textbf{0.885} & 0.378          & 0.015          \\
                              & Complt & 0.566          & \textbf{0.573} & 0.239          & 0.396          & \textbf{0.457} & 0.265          \\ \cline{2-8} 
\multirow{3}{*}{Kaggle ASM}   & V-M    & \textbf{0.528} & 0.235          & 0.014          & \textbf{0.562} & 0.531          & 0.366          \\
                              & Homog  & \textbf{0.550} & 0.176          & 0.007          & \textbf{0.911} & 0.599          & 0.291          \\
                              & Complt & \textbf{0.508} & 0.351          & 0.257          & 0.407          & 0.477          & \textbf{0.495} \\ \cline{2-8} 
\multirow{3}{*}{Drebin APK}   & V-M    & \textbf{0.307} & 0.219          & 0.095          & \textbf{0.412} & 0.326          & 0.389          \\
                              & Homog  & \textbf{0.296} & 0.172          & 0.054          & \textbf{0.566} & 0.313          & 0.333          \\
                              & Complt & 0.319          & 0.303          & \textbf{0.375} & 0.323          & 0.340          & \textbf{0.468} \\ \cline{2-8} 
\multirow{3}{*}{Drebin TAR}   & V-M    & \textbf{0.403} & 0.248          & 0.065          & \textbf{0.508} & 0.478          & 0.386          \\
                              & Homog  & \textbf{0.416} & 0.177          & 0.036          & \textbf{0.754} & 0.503          & 0.332          \\
                              & Complt & 0.391          & \textbf{0.413} & 0.335          & 0.383          & 0.455          & \textbf{0.460} \\ \cline{2-8} 
\multirow{3}{*}{VS20F}        & V-M    & \textbf{0.353} & 0.009          & 0.009          & \textbf{0.449} & 0.204          & 0.056           \\
                              & Homog  & \textbf{0.328} & 0.005          & 0.005          & \textbf{0.562} & 0.137          & 0.030          \\
                              & Complt & \textbf{0.381} & 0.249          & 0.221          & 0.374          & \textbf{0.400} & 0.378          \\ 
\bottomrule
\end{tabular}
}
\end{table}
BWMD is
the only method that can leverage the k-Means algorithm, and we use Hamerly's variant because it avoids redundant computation while returning the exact same results\cite{Hamerly2010}. For LZJD and BitShred we use Average-Link clustering using a  fast $\bigO(n^2)$ algorithm \cite{Mullner2011}. While the original BitShred paper used Single-Link, we found Average link provided the best results across all metrics for both BitShred and LZJD.  The results are shown in \autoref{tbl:cluster_results}, where we can see BWMD dominates LZJD and BitShred by our most important metrics, Homogeneity and V-Measure\footnote{Not included due to space limitations, BWMD also dominates by Normalized Mutual Information and Adjusted Rand Index.}. BWMD's advantage in this regard is often dramatic. For example, BWMD scores $2.34\times$ better on Homogeneity compared to LZJD when $k=10\cdot C$ on the Kaggle bytes dataset, and $59\times$ better than BitShred. While BWMD does not always perform best by the Completeness metric, it is always competitive with the best scoring method, which is why BWMD dominates by V-Measure 
. The results overall clearly indicate that BWMD provides the best clusterings across multiple datasets, of different encodings, and different numbers of clusters, showing the flexibility of the compression-based approach. 

Because BWMD can leverage the k-means concept and the many efficient algorithms for its computation, it is also the most scalable for these methods. 
LZJD and BitShred are inherently limited by the $\bigO(n^2)$ lower-bound complexity of hierarchical clustering
\footnote{Other alternatives, like k-medoids, have similar $\bigO((n-k)^2 k)$ complexity.}
. For example, BWMD took only 27 minutes to over-cluster the 160,000 files in the VS20F training set, the largest under consideration. This is $17.3\times$ faster than LZJD which took 7.76 hours, and $54.6\times$ faster than Bitshred at just over a day. 

\section{Conclusion} \label{sec:conclusion}

We have developed and introduced the Burrows Wheeler Markov Distance (BWMD), a new distance metric
inspired by the Burrow Wheeler Transform.
A theoretical analysis has shown several ways in which BWMD has better behavior, which is confirmed by showing new abilities for clustering DNA sequences that prior methods could not handle. For malware clustering, we have shown BWMD considerably outperforms prior methods in both speed and accuracy, and BWMD is the only byte-based method which can achieve sub-linear search scaling on larger corpora. 

\bibliography{Mendeley}
\bibliographystyle{abbrvnat}

\clearpage

\begin{appendices}

\section{Compression Distances With High Entropy}

As a final analysis, we will look at the behavior of EBWT, LZJD, and our new BWMD in cases where input strings have high entropy. Such situations are common in malware classification, due to the use of "packing", where a malware author will use compression and/or encryption to either reduce a file's size or hide its contents \cite{Martignoni2007}. 

\subsection{Similarity of Two Random Files}

First we will consider the behavior when measuring the distance between two files, both of which are completely uniform random strings of characters, and so have a maximal entropy $-\sum_{\forall s \in \Sigma} -p_s \log_{|\Sigma|}(p_s) = 1$, where $p_s$ is the probability of observing character $s$ from the alphabet $\Sigma$. For an empirical perspective, in \autoref{fig:contour_rand_files_dist} we show a contour plot of the normalized distance (i.e., scaled to be in the range $[0,1]$) over an alphabet $|\Sigma|=32$. On the x and y axis are the lengths of the random files being compared, and the value plotted is the average normalized distance over 30 trials of different random sequences for each trial. 
In the plot, dark blue indicates a near zero distance, and dark red a distance near one. 

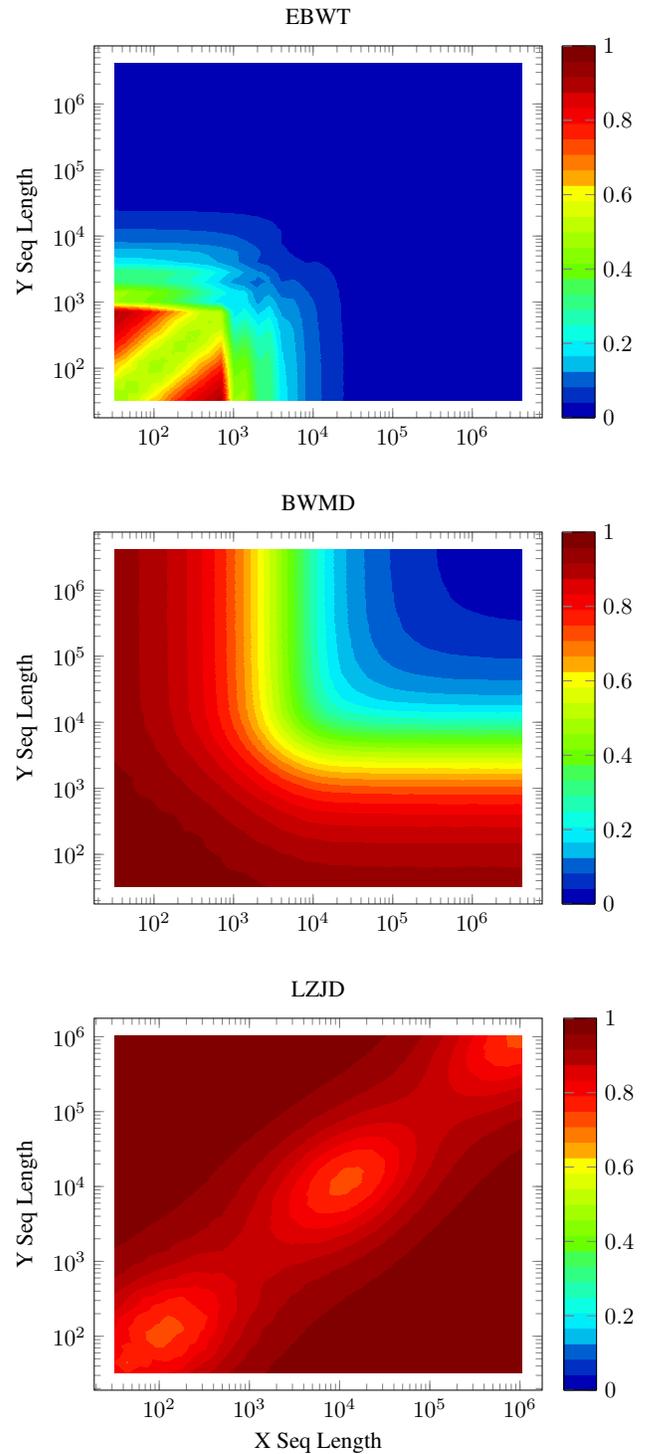
\begin{figure}[!htb]
\centering
\begin{adjustbox}{max size={\columnwidth}{\textheight}}
  \begin{tikzpicture}
  \begin{groupplot}[
      group style={
        	group name=myplot, group size=1 by 3,
			vertical sep=1.75cm,
        },
        enlarge x limits=true,
      ]
    \centering
	\nextgroupplot[
		title=EBWT,
		ylabel=Y Seq Length,
        ymode = log,
        xmode = log,
        colormap/bluered,
        colorbar sampled,
		view={0}{90},
		point meta min=0.0,
        point meta max=1.0,
        enlarge x limits=0.05,
		enlarge y limits=0.05,
	]
	\addplot[
        point meta=\thisrow{z},
        contour filled={
            number=28,
        }
		] 
		table {ebwt_rand.dat};
		
	\nextgroupplot[
		title=BWMD,
		ylabel=Y Seq Length,
        ymode = log,
        xmode = log,
        colormap/bluered,
        colorbar sampled,
		view={0}{90},
		point meta min=0.0,
        point meta max=1.0,
        enlarge x limits=0.05,
		enlarge y limits=0.05,
	]
	\addplot[
        point meta=\thisrow{z},
        contour filled={
            number=28,
        }
		] 
		table {bwmd_rand.dat};
		
	\nextgroupplot[
		title=LZJD,
		xlabel=X Seq Length,
		ylabel=Y Seq Length,
        ymode = log,
        xmode = log,
        colormap/bluered,
        colorbar sampled,
		view={0}{90},
		point meta min=0.0,
        point meta max=1.0,
        enlarge x limits=0.05,
		enlarge y limits=0.05,
	]
	\addplot[
        point meta=\thisrow{z},
        contour filled={
            number=28,
        }
		] 
		table {lzjd_rand.dat};
		

    \end{groupplot}

  \end{tikzpicture}
\end{adjustbox}
  \caption{Distance between two random files of different lengths.  }
  \label{fig:contour_rand_files_dist}
\end{figure}

For EBWT, the intuition is that random content will result in a random ordering in the EBWT merge step, and thus have a few source repetitions, and ultimately produce a small normalized distance (i.e., $ebwt(u,v)/(|u|+|v|-2)$). We can see that EBWT follows this behavior, but has an unusual and undesirable pattern of non-zero values for smaller sequences ($|u| \leq 10^3$). Along the diagonal of equal-length files, a value of $\approx 0.4$ occurs, but then increases on the off diagonal, before decreasing again once the size $|u| > 10^3$. 

For BWMD, this analysis is more grounded. We know that for a random sequence, the transformed BWT has nothing it can compress, and will result in a different  permutation of the still random string. This means the distribution of the prefix to every substring will not change (i.e., the preceding character is random, which is true by definition), and so we expect the transition vector $x$ in step 3 of the BWMD to converge toward the uniform distribution as the size of the input string increases. So in the limit, we know two random files will receive a distance near zero. However, if one (or both) of the random  files are sufficiently small, a large distance will occur due to a lack of convergence toward the uniform distributions. The rate of change will depend on the alphabet size $|\Sigma|$, and can be described as the expected value of the differences between transformed multinomial distributions. 
While bounding this value mathematically is difficult, we can see in \autoref{fig:contour_rand_files_dist} that this expectation holds true empirically. 

Because LZJD measures the similarity of compression dictionaries from the Lempel Ziv algorithm, it was assumed that distances would become smaller as the size of a random sequence increased, due to the accumulation of all smallest substrings \cite{raff_lzjd_2017}. We see in practice that this does not quite occur. For sequences of significantly different length, LZJD returns a distance near 1, in part due to the size difference forcing the intersection to be a small portion of the union of two files. When two random sequences are of a more similar length, we see that the LZJD distance does begin to drop to nearly $\approx 0.6$, but still remains relatively high. This makes sense due to the random component of the input strings. As the set of all shortest substrings is constructed, consider the scenario where the compression dictionary has all substrings of length $\alpha$, but has only processed half of the total sequence. For the remaining half, we will obtain an essentially random sampling of strings of length $\alpha+1$, of which there are $|\Sigma|$ as many strings of increased length (one for every character in the alphabet $\Sigma$, times every string of length $\alpha$). If two strings are of the same length, the likelihood of seeing the same strings of length $\alpha+1$ then becomes small, increasing the distance, while all strings of size $\leq \alpha$ increase the similarity and reduce distance. 

Between BWMD and LZJD, we have two opposing styles of measuring the distance between two high-entropy sequences. As was noted in \cite{raff_lzjd_2017}, one could argue that two random sequences should have high distance because they share no underlying information (this is how NID would behave), or lower distance because both sequences share the property of "looking random". While we argue that EBWT has a more confusing and inferior behavior, a preference to BWMD or LZJD's behavior may be task-specific. 

\subsection{High Entropy vs Low Entropy}

An important distinction in behaviors on random sequences is when a uniform random sequence (max entropy=1) is compared to a second sequence with a different level of entropy. For BWMD, we can derive intuition about its behavior in this case from work on analyzing the BWT compression algorithm. In particular, it has been shown that the BWT's efficiency at compression of $k$'th order combinations over the alphabet $\Sigma$ is bounded by $k$'th order empirical entropy \cite{Manzini:2001:ABT:382780.382782}. Thus we should anticipate that BWMD will produce differing Markov transition matrices, and thus differing distances, as the entropy of the second sequence changes. We confirm this empirically in \autoref{fig:entropy_change}\footnote{The results are shown only for the case of two sequences of the same specific length. The general trend, that BWMD's distance changes as the entropy of one sequence changes, holds for other sequence length combinations as well --- but such results are excluded due to space limitations. }. 

\begin{figure}[!htb]
\begin{center}
\centering
\begin{tikzpicture}[]
\begin{axis}[
    xlabel=Target Sequence Entropy,
    ylabel=Average Distance,
    legend pos=south east,
    legend style={at={(0.05,0.45)},anchor=west},
    legend columns=1,
    height=4.5cm,
    width=\columnwidth,
    ]
    
    \addplot+[mark=none,line width=2.0pt] table [x=Entropy, y=EBWT, ] {entropy_var.dat};
    
    \addplot+[mark=none,line width=2.0pt] table [x=Entropy, y=BWMD, ] {entropy_var.dat};
    
    \addplot+[mark=none,line width=2.0pt] table [x=Entropy, y=LZJD, ] {entropy_var.dat};

    \legend{EBWT, BWMD, LZJD} 

\end{axis}

\end{tikzpicture}
\end{center}
\caption{Average distance (y-axis) between a primary sequence of length $10^5$ that is uniformly random (entropy=1) and a second sequence of a specific entropy between 0 and 1 (x-axis), and also of length $10^5$. }
\label{fig:entropy_change}
\end{figure}
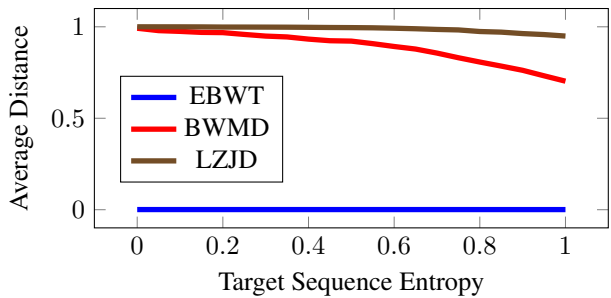

We can see that LZJD and EBWT do not share this desirable property of differing distance with differing entropy. Given a single sequence of high entropy, LZJD will deem almost all other entropy sequences equidistantly far, and EBWT all other entropy sequences equidistantly near. In the case of LZJD, this results in behavior effectively equivalent to when the dimension of a problem continues to increase, and all other points becoming increasing equidistant to each other \cite{Beyer:1999:NNM:645503.656271}. We believe BWMD's ability to delineate degrees of similarity/distance based on relative entropy levels is an important factor toward its ability to provide better clustering results compared to LZJD, which will be shown later in \autoref{sec:malware}.

\end{appendices}

\end{document}